\documentclass[a4paper,11pt]{article}

\usepackage{draft_style}

\begin{document}
	
\title{\textbf{Spreading of an infectious disease between different locations}\thanks{We acknowledge funding from the Italian Ministry of Education Progetti di Rilevante Interesse Nazionale (PRIN) grant 2015592CTH. We are grateful to the Italian Ministry of Health, Istituto Zooprofilattico Sperimentale dell'Abruzzo e del Molise ``G. Caporale'' and, in particular, to Luigi Possenti and Diana Palma for their help with the data. For their helpful comments we would like to thank  Alberto Alesina (and the participants to his reading group at Bocconi University), Leonardo Boncinelli, Simone D'Alessandro, Jakob Grazzini, Kenan Huremovi\'{c} and  Roberto Patuelli.}}

\date{This version: December 2018 \\
	\href{https://www.dropbox.com/s/9uhze4nbbsis1s7/MPR_spreading_epidemics_locations.pdf?dl=0}{[here you find an updated version]}}

\author[a]{Alessio Muscillo}
\author[b]{Paolo Pin}
\author[a,c]{Tiziano Razzolini}

\affil[a]{Department of Economics and Statistics, University of Siena}
\affil[b]{Department of Decision Sciences, IGIER and BIDSA, Bocconi University}
\affil[c]{IZA and LABOR}

\maketitle

{
\begin{abstract}


The endogenous adaptation of agents, that may adjust their local contact network in response to the risk of being infected, can have the perverse effect of increasing the overall systemic infectiveness of a disease.
We study a dynamical model over two geographically distinct but interacting locations, to better understand theoretically the mechanism at play.
Moreover, we provide empirical motivation from the Italian National Bovine Database, for the period 2006-2013.

\end{abstract}
}

\noindent \textbf{JEL\ classification codes}: C32, D83, I12	



\section{Introduction}

Connections between individuals are beneficial because they enable the exchange of goods and resources.
However, they  are also the means through which diseases and shocks may spread in a society, making it vulnerable to hazards and menaces. Due to the advances in virtual and physical communications, understanding this tradeoff has become increasingly more necessary as well as complicated.

We study how the spread of an infection evolves in a population adopting self-protecting behavioral responses which, in turn, affect the evolution of the epidemics. We focus on two interplaying mechanisms: First, how does the contact network influence the evolution of the disease by only allowing contagion via existing contacts? Second, how is the network itself endogenously modified by the behavioral response triggered by the risk perception?

In the context of a simple two-location model, we obtain analytical results from a system of ordinary differential equations. This very stylized cross-location interaction may generate complex dynamics in terms of the co-evolution of the coupled mechanism constituted by the contact network and the infection spread.
The main question concerns how this stylized globalization affects the 2-location systemic resistance with respect to shocks in the infection rates. Is the coupled system more resistant to shocks than 2 separated and autarkic single locations? What happen when shocks simultaneously hit both locations?


We assume that each location has a limited recovery ability from the disease (e.g. limited hospitalization capacity for quarantine) and that sudden outbreaks in the infection rate (also called \textit{infection shocks} hereafter) occur exogenously and abruptly.
After being hit by such an initial infection shock, the evolution of the disease (and the effectiveness of the recovery measures) is observed over time.
Small shocks are better controlled when the two locations are connected together than when they are isolated and autarkic: in fact, infected individuals who (out)flow from the most infected location to the least one, are treated in both locations and this contributes to diluting and reducing the epidemic overall.
On the contrary, a large shock, even if initially concentrated in only one location, may end up infecting completely both locations, thus putting at risk the entire systemic resistance to contagion.

In terms of policy implications, given the characteristics of the disease under study (e.g. contagiousness, type of infection shocks) and given the ease of connection between the locations (e.g. autarky or globalization), the resistance of the whole system to infections depends on the resources allocated for recovery measures. 
As the system becomes more and more globalized by facilitating connections between distant locations, it becomes also more resistant to small infection shocks but, conversely, it also becomes more exposed to large shocks. Moreover, the relative advantage (in terms of systemic resistance) of a globalized world with respect to a system of autarkic or isolated locations becomes higher as the amount of resources dedicated to recovery measures increases, because of complementarities established between the two locations.

\subsection*{Related literature}
\label{sec_literature}

Epidemiological models have been studied for decades, starting from the seminal Kermack-McKendrick compartmental models that go back to the 1920s and 1930s \citep{allen2008mathematical}. In recent years, more attention has been devoted to incorporating agents' behavioral response and awareness to the concurrent evolution of the disease in the population \citep{funk2010modelling, fenichel2011adaptive, poletti2012risk}.
Moreover, because of the facility through which interconnections and interdependencies are established in a globalized world, better models need also to account for different individuals' mobility patterns \citep{brauer2001models, wang2004epidemic, manfredi2013modeling}.

In a literature closer to economics and to the social sciences, some theoretical works deal with strategic vaccination or with the adoption of different defensive mechanisms which may depend on the connections of the individuals \citep{galeotti2013strategic, galeotti2015diffusion, goyal2015interaction}. 
However, with respect to this paper, other works share the same motivation dealing with diseases spreading through trade connections \citep{horan2015managing} or approach a somehow similar problem with different methodologies \citep{reluga2009sis}.

\subsection*{Roadmap}

The paper continues as follows. Firstly, we provide novel and fundamental empirical motivations for our analysis (Section \ref{sec_motivation}).
We then move to our model, with the introductory case of a single location (Section \ref{sec_1island}), the description and the results of the main model (Section \ref{sec_model}) and its comparative statics with respect to exogenous shocks (Section \ref{sec_discussion}).
Section \ref{sec_conclusion} concludes the main part. Additionally, the appendix includes an extension of the empirical exercise done in the motivation section (\ref{sec: econometrics_appendix}), the proofs omitted in Sections \ref{sec_1island}, \ref{sec_model} and \ref{sec_discussion} (\ref{sec: proofs}) and the mathematical analysis of the linear case (\ref{app_linear}) which ends with the approximation of the basins of attraction of the equilibria used for the comparative statics analysis (\ref{app_comparative}).

\section{Motivation for  our exercise}
\label{sec_motivation}

Before going on to the description of the model and of its contribution to the existing literature, we provide two strong motivations for our work.
The first one comes from a novel dataset, while the other is a relevant application.

\subsection*{Livestock trading: infections and long-range connections}

We perform an empirical analysis of trade flows of bovines in the Italian territory using the Italian National Bovine Database (\textit{Anagrafe Nazionale Bovina}). This dataset has been created by the Italian Ministry of Health after the outbreak of the  Bovine Spongiform Encephalopathy (BSE) in accordance with the European Economic Community Council Directive 92/102/EEC of 1992. The Directive imposed to all member states to identify each bovine using ear tags and to follow all its movements, from birth until death, through all holdings (farms, assembly centers, slaughterhouses, markets, staging points, pastures, foreign countries of origin) in the national territory.

For each movement of bovines, we have information on the location (e.g. municipality) of approximately 220,000 origin and destination premises,  90\% of which are farms. The dataset records the exact date of all these movements between 2006-2013 and contains information on the stock of animals in each holding.\footnote{Information on the stock is recorded on the same date of the movement, before any inflows or outflows have occurred. With the latter information and the data on movements we compute the stock of animals at the beginning of each quarter.} Information on trade flows have been merged with the SIMAN database (\textit{Sistema Informatico Malattie Animali}) which registers the diseases occurred in each holding \citep{iannetti2014integrated, calistri2013systems}.
We refer also to \cite{mprs2018import} for more detailed information.

In our analysis, we will focus on outflows from farms in each quarter, from 2007Q1 to 2013Q4.\footnote{Since 2006 was the first year of introduction of the tracking system we start using the data from the first quarter of 2007, when the running-in phase was over.}
We are mainly interested in determining whether the occurrence of a disease in farm $i$ at time $t-1$ affects the distances of trade flows originated from farm $i$ at time $t$.

Figure \ref{fig: histd} shows the histogram of distances for all the flows under analysis. The median value is about 17 kilometers whereas the average value and the 75th percentile are respectively 43 and 41 kilometers. About 10\% of the movements occur between farms in the same municipality, thus producing a  point mass at zero in the distribution of distance.

Since some staging points or assembly centers can be misclassified as farms, we retain in the sample only the holdings with a value of the stock smaller than the 99th percentile, which is equal to 954.\footnote{The 99th percentile is twice the size of the 95th percentile, the triple of the 90th percentile and about 14 times the size of the median. The empirical estimates from the unrestricted sample are qualitatively and quantitatively similar.} Our final sample consists of 117,758 farms originating 1,541,370 trade flows towards other farms. 

\begin{figure}[htb!]
	\centering
	\caption{\textbf{Distances of trade flows}}
	\includegraphics[width=0.65\textwidth]{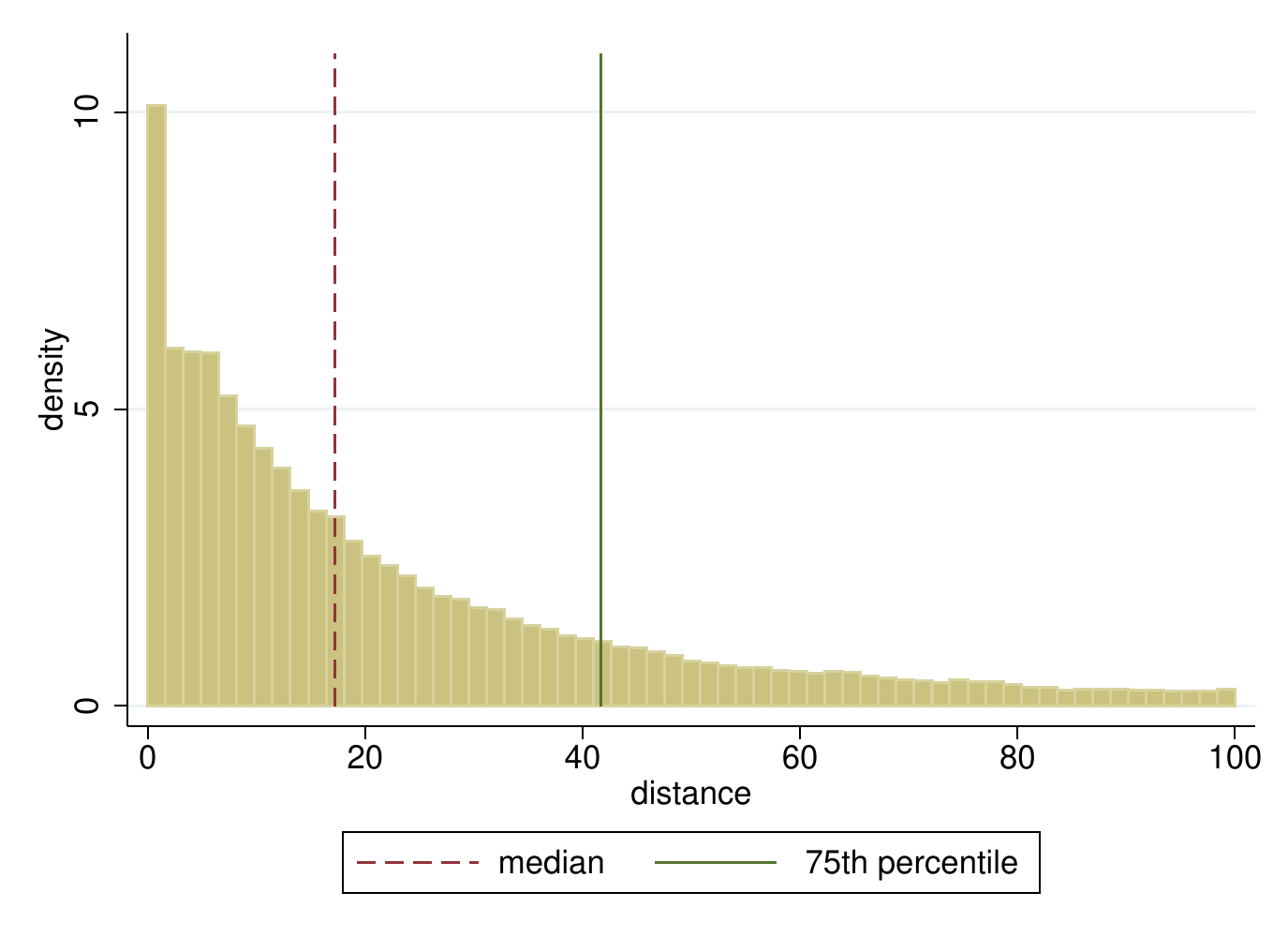}
	\label{fig: histd}
\caption*{\footnotesize{The figure displays the histogram of distances (in kilometers) shorter than 100km of the trade flows occurred in the period 2007-2013. Descriptive statistics of distances are shown in Table \ref{tab: destat}.}}
\end{figure}

We thus estimate different regressions with trade flow distance as the dependent variable using the following specification:
\begin{equation}
\text{Distance}_{it} = \beta_0 + \beta_1\ \text{Positive}_{i,t-1} + \beta_2\ \text{Stock}_{it} +  \text{Region}_g + \theta_t + \alpha_i + \epsilon_{it}
\end{equation}
\begin{equation*}
 t=2007Q1,\dots,2013Q4; \qquad g=1,\dots,20
\end{equation*}

where the dummy $\text{Positive}_{i,t-1}$ is equal to one if in the previous quarter farm $i$ has registered at least one disease; $\text{Stock}_{it}$ measures the number of bovines in farm $i$ at the beginning of quarter $t$. The variables $\text{Region}_g$ are regional dummies and $ \theta_t$ are quarter dummies, while $\alpha_i$ are farm-specific time invariant effects used only in fixed effect regressions.

Table \ref{tab: destat} reports descriptive statistics of the dependent variable and the main regressors. In the 2007-2013 period we observe 265 diseases which represent 0.02\% of the observations (i.e. all movements) used in the empirical analysis. 

\begin{table}[htb!]
\caption{Descriptive statistics}
\label{tab: destat}
\centering
\begin{tabular}{l@{\hspace*{5mm}}ccccc}
\hline\hline \vspace{-4mm}\\
			& Mean & St. Dev. &	Median & Min & Max\\
\hline \vspace{-4mm}\\
Distance 	& 42.80 &  92.18 & 17.08& 0 & 1,291\\
Positive 	& 0.0002& 0.0131&0 &0 &1\\
Stock  		& 119.36 &142.87 &66& 1 & 954 \\
\hline \vspace{-4mm}\\
Observations & 1,541,370&&&&\\
N. farms & 117,758 &&&&\\
\hline\hline
\end{tabular}
\caption*{\footnotesize The sample excludes farms with a value of the stock of bovines greater than 954.}
\end{table}

The empirical results are shown in Table \ref{tab: empres}. Column 1 reports estimates from a Probit model where the dependent variable is set equal to one when the distance is larger than 41 kilometers (i.e the 75th percentile.)
The estimated coefficient of the dummy Positive is statistically significant at 5\% level. The marginal effect is 0.05 which means that presence at $t-1$ of a sick animal increases by 20\%  the probability of sending bovines in quarter $t$ to farms distant more than 41 kilometers.

The estimates from standard OLS regressions are shown in column 2.  The  coefficient of the dummy Positive indicates that a disease in the past increases the distance of trade by about 19 kilometers.
In column 3 we report  the estimated effect from a Tobit to take into account the censored nature of distance which has a point mass at zero. The estimated effect in this case is equal to 19 kilometers. 
Finally, in column 4, we show the estimates from a panel regression to control for farm-specific fixed effects.
The estimated coefficient of the dummy Positive indicates that a registered disease at $t-1$ increases the distance by about 10 kilometers. All estimation techniques provide evidence of a positive and significant effect of past diseases on the current value of distance.

The above-mentioned results are, however, conducted on the sample of farms that exhibit non-zero movements of bovines. We have performed a robustness check to take into account possible selection effects. We have estimated a bivariate Tobit model with sample selection by maximum likelihood estimation. The results show evidence of a weak negative selection: farms that  are more likely to be active tend to send bovines to closer location. Moreover, farms that have experienced a disease in $t-1$ are less likely to be active in quarter $t$. Nevertheless, the estimated effect of the dummy Positive on distance is again equal to 19 kilometers.
Estimation results and the identification strategy adopted are shown in \ref{sec: econometrics_appendix}.

\begin{table}[htb!]
\caption{Empirical Results}
\label{tab: empres}
\centering
\begin{tabular}{l@{\hspace*{5mm}}cccc}
\hline\hline \vspace{-4mm}\\		
& (1) 		& (2) 		& (3) 			& (4) 		\\
& Probit	& OLS 		& Tobit 		& Panel FE 	\\
\hline \vspace{-4mm}\\
$\text{Positive}_{i,t-1}$ 	& 0.1628** 	& 18.7323***& 19.1025***& 10.9276**\\
							&(0.081)	&(5.420)	&(5.587)	&(4.735)\vspace{2mm}\\
$\text{Stock}_{i,t}$ 	&  0.0010***& 0.0821*** & 0.0846***	& 0.0490***\\
						& (0.000) 	& (0.001) 	& (0.001) 	& (0.002)\vspace{2mm}\\
Constant 				& -1.2508***	& 9.3844*** &	7.4312*** 	&	31.9721***\\
						&(0.010) 		&(0.553) 	& (0.570) 		& (0.406)\vspace{2mm}\\
$\sigma$ 				& 	& 	& 90.528*** &\\
       					&  	&  	& ( 0.0528) & \\ 
\hline \vspace{-4mm}\\			
\small{$\partial\,\text{Pr}(Y_i =1) /\partial\,\text{Positive}$}  & 0.0524**&&&\\
  & (0.0271)&&&\\
\hline \vspace{-4mm}\\
Observations	& 1,541,370 &1,541,370 	&	1,541,370 &	1,541,370\\
Log likelihood 	&-838,742.52&			& -8,819,135.6 &\\
Adj. R-squared	& 			& 0.0845 	&		&0.3927\\
\hline\hline
\end{tabular}
\caption*{\footnotesize All regressions control for time dummies. Estimations in columns 1,2 and 3 control for regional fixed effects. Panel fixed effect estimation in column 4 controls for farm-specific effects. Standard errors clustered at the farm level are shown in parenthesis. Asterisks mean: *** significant at 1\%, ** significant at 5\%,* significant at 10\%.}
\end{table}

\subsection*{The 2014 Ebola outbreak}

The theoretical study of infection dynamics when the (endogenous) behavior of patients increases infections has potentially enormous applications.
Such a mechanism has been at play (and possibly at the origin of) some disastrous epidemic events, like the complex case of the \emph{Zaire ebolavirus} epidemic that affected West Africa since approximately December 2013 (see Figure \ref{fig:map}).\footnote{See \cite{chowell2014transmission} for a detailed review; see also \cite{thomas2015economicebola} and the website of the World Bank for some estimates of the damages done to the economies of some African countries: 
\url{http://www.worldbank.org/en/region/afr/publication/ebola-economic-analysis-ebola-long-term-economic-impact-could-be-devastating} and also \url{http://www.worldbank.org/en/region/afr/publication/the-economic-impact-of-the-2014-ebola-epidemic-short-and-medium-term-estimates-for-west-africa}.}
A particularly dangerous situation can occur when contagious individuals are expelled from their villages and are able to reach big towns, or even other countries; or if they voluntary travel to other countries when they are sick, to avoid social stigma or to obtain a better treatment.

The World Health Organization (WHO) reports that in 2014, during the Ebola outbreak in West Africa:\footnote{\url{http://www.who.int/csr/disease/ebola/one-year-report/factors/en/}}
\begin{quotation}
``\dots as the situation in one country began to improve, it attracted patients from neighboring countries seeking unoccupied treatment beds, thus reigniting transmission chains. In other words, as long as one country experienced intense transmission other countries remained at risk, no matter how strong their own response measures had been.''
\end{quotation}
Directly quoted from the web site of the WHO:\footnote{\textit{Ibidem.}}
\begin{quote}
``Countries in equatorial Africa have experienced Ebola outbreaks for nearly four decades. 
[...] In those outbreaks, geography aided containment. 
[...] In West Africa [which had never experienced an Ebola outbreak], entire villages have been abandoned after community-wide spread killed or infected many residents and fear caused others to flee.
[...] West Africa is characterized by a high degree of population movement across exceptionally porous borders. Recent studies estimate that population mobility in these countries is seven times higher than elsewhere in the world. 
[...] Population mobility created two significant impediments to control. 
[...] [C]ross-border contact tracing is difficult. Populations readily cross porous borders but outbreak responders do not. 

The importation of Ebola into Lagos, Nigeria on 20 July and Dallas, Texas on 30 September [2014] marked the first times that the virus entered a new country via air travelers. These events theoretically placed every city with an international airport at risk of an imported case. The imported cases, which provoked intense media coverage and public anxiety, brought home the reality that all countries are at some degree of risk as long as intense virus transmission is occurring anywhere in the world –- especially given the radically increased interdependence and interconnectedness that characterize this century.''
\end{quote}

\begin{figure}[htb!]
    \centering
    \includegraphics[width=0.75\textwidth]{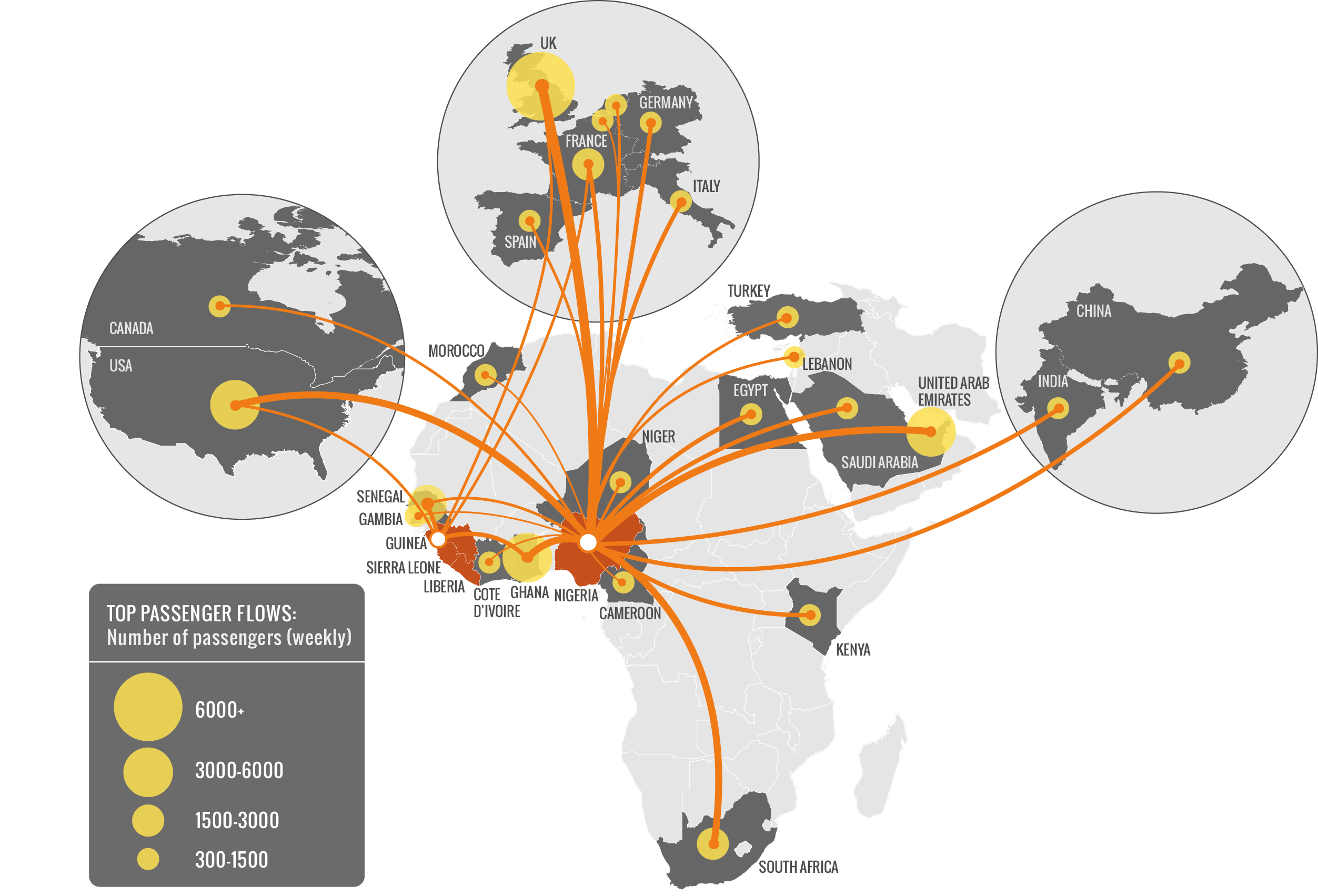}
    \caption{\small Air traffic connections from West African countries to the rest of the world. Source \cite{gomes2014assessing}. See also \cite{halloran2014ebola}.}
    \label{fig:map}
\end{figure}

\section{The single-location model: the building block}\label{single-locationmodel}
\label{sec_1island}

As a warm-up exercise, in this section we develop the building block of the model. We define a system constituted by a single location and describe how the infection evolves in it, as time passes. The dynamics is kept explicitly abstract and simple on purpose, and this has one main reason: the so-defined 1-location system is able to recover from small shocks, but unable to do so in case of large shocks (to be defined shortly). This, in turn, will allow us in the next Section \ref{sec_model} to consider two such systems interacting with each other and, then, evaluate what will be the effects on this whole 2-location system, in terms of resistance to shocks.

Consider a population of agents living in one location and susceptible to the infection from a transmittable disease, which can spread through personal contacts with other agents. Following the motivations from Section \ref{sec_motivation}, the intuitive idea is that agents trade with each other and meet in pair. These meetings, however, are also the mean through which the disease may spread.

Let $x(t)$ denote the fraction of infected individuals at time $t$. The evolution of this fraction is ruled by the following differential equation, used for example in ecological economics as a development from the classic Bass model \citep{bass1969model, dalessandro2007nonlinear}:
\begin{equation}\label{SI_cubic}
\frac{\de }{\de t} x(t) = \nu x(t)(1-x(t))(x(t)-q),
\end{equation}
where $\nu \in (0,1)$ is a parameter representing the \textit{contagiousness} of the disease and $q \in (0,1)$ is a parameter which measures the capacity of the system to control the disease, which we may call \textit{quarantine}. 
More specifically, we think of $q$ as the quantity of resources allocated to hospitalize infected individuals as well as to other disease-control measures. We consider these resources fixed and exogenous, meaning that they can change over a longer time-scale with respect to that of the evolution of the disease.

\begin{remark}
Equation (\ref{SI_cubic}) can be seen as modified susceptible-infected model: we take the probability that an infected individual meets a susceptible one, i.e. $x(1-x)$, and that this meeting results in a new infection with probability $\nu$. We then multiply this by a factor $(x-q)$ which modifies the sign of the flow of infected depending on whether the fraction of infectives exceeds or not the quarantine threshold $q$.
\end{remark}

\begin{proposition}
\label{prop_equilibria_1location}
	The dynamical system \emph{(\ref{SI_cubic})} has 3 critical points:
	\begin{itemize}
		\item	the asymptotically stable, \textit{disease-free equilibrium} $x=0$;
		\item	an unstable equilibrium $x=q$;
		\item	the asymptotically stable, \emph{endemic equilibrium} $x=1$.
	\end{itemize}
Consequently, the interval $[0,q)\subset \R$ is the basin of attraction of $x=0$, while $(q,1]$ is the basin of attraction of $x=1$.
\end{proposition}
\begin{proof}
	See \ref{sec: proofs}.
\end{proof}

%


\begin{remark}
	In ecology, this dynamics sometimes describes the evolution of a species over time \citep{dalessandro2007nonlinear}. In this context, the analogy is that the species we are considering is a bacterium causing the infective disease. The threshold $q$ represents the \textit{critical mass} of infections that the species has to reach and exceed in order to survive: when there are not enough infected individuals, the species cannot proliferate and propagate any more and, eventually, the epidemic dies out. 
	Lastly, it is worth noticing that underlying assumption in the susceptible-infected model is the so-called \emph{homogeneous mixing}: meetings among individuals are random, according to their relative proportion in the whole population. This is an assumption that we maintain here.
\end{remark}

\paragraph{Resistance to shocks \& policy}
We define a \emph{shock} as follows: suppose that at time $t=0$ there is a sudden and exogenous variation in the infection rate such that $x(0) = x_0 \in [0,1]$. This initial fraction $x_0$ of infected is what we will call shock.

If the shock is $x_0 < q$, i.e. below the threshold, then the system will (asymptotically) return to the disease-free equilibrium, whereas if the shock is larger than $q$, then the dynamic will converge toward 0, where the whole population is infected. If the shocks are assumed uniformly distributed over $[0,1]$, then $q$ quantifies the ability of the system to recover from a shock.

The policy implication here is then straightforward: the more resources can be allocated to control the disease (i.e. the larger $q$ is), the more the system will able to recover from larger shocks in the infection rate (i.e. the larger will be the basin of attraction of 0).


\begin{figure}[htbp]
	\centering
	\caption{\textbf{Single-location dynamics}}
	\includegraphics[width=.4\textwidth]{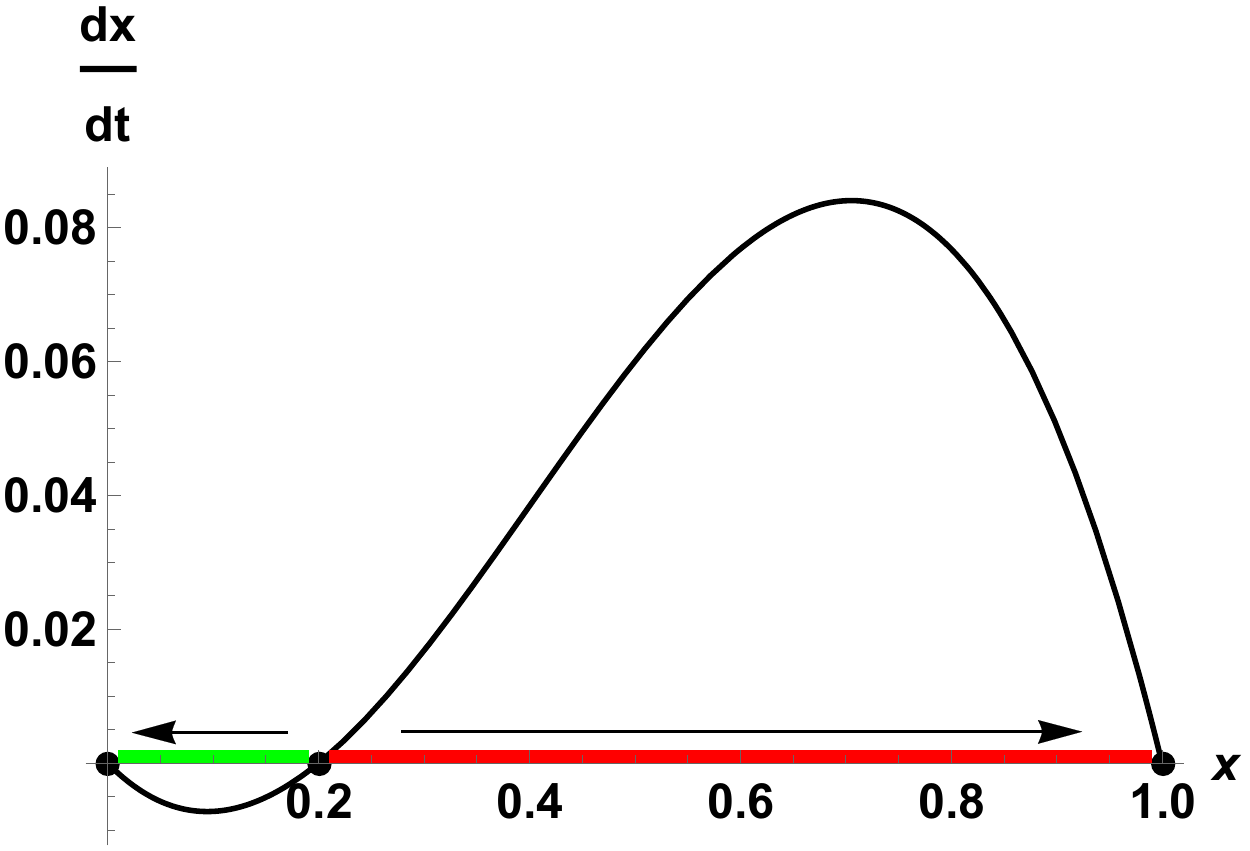}
	\vspace{2mm}
	\caption*{\small Dynamics defined by equation (\ref{SI_cubic}), where the parameters are set at $q=0.2$ and $\nu=0.8$. The curve represents the right-hand side of equation (\ref{SI_cubic}) and the dots are the critical points $0$, $q$ and $1$. The unstable equilibrium $x=q$ acts as a threshold separating the basins of attraction of the two asymptotically stable equilibria $x=0$ (in green) and $x=1$ (in red). Small exogenous shocks, i.e. below $q$, are absorbed, whereas shocks larger than $q$ lead to a fully infected system.}
	\label{SIcubic}
\end{figure}

\section{The 2-location model}
\label{sec_model}

Starting from the conclusion of the previous section, we now extend our analysis: trades and meetings will take place both within and across two geographically distinct locations (also called islands or countries hereafter) and, consequently, the same will happen to the spread of the disease.
To express the incentives of economic agents, who can choose whether to interact with other agents within their own location or in the other location, we stick to the interpretation of trading and, therefore, speak also of prices.
We think that the economic intuition remains the same also in other situations where prices are less explicit (as for the Ebola example of Section \ref{sec_motivation}), since things could still be modeled in terms of higher and heterogeneous costs for interaction with distant locations.

Specifically, we consider 2 locations both populated by interacting agents, e.g. farmers who are trading cattle. Agents benefit from interacting with each other but, since there may be a (latent) disease spreading, this potential benefit decreases as the infection prevalence increases. This accounts for the risk of becoming infected and the reduction in performance that diseased cattle experience (e.g. slower growth, death).
In the attempt to avoid contagion and risks, the agents of one location may be willing to interact with other agents in the other location, even if to do so they have to pay a higher cost related to this long-range interaction (e.g. export costs, trade barriers).

We restrict our attention to two identical and symmetric locations, where agents are homogeneous and identical in all aspects but in the export cost. 
In particular, different agents of the same location are assigned different costs to export to the other location and, intuitively, this may be reflecting different geographical proximity, facility in the contacts with a foreign country, etc.
One key aspect is that using identical locations and identical agents is a normalization that can help guarantee that any variation in the fragility of the coupled system is due to the cross-country connection structure rather than to differences in other characteristics.

\subsection{Specification}

Let $A$ and $B$ denote two populations of agents living for an infinite time horizon and let $a$ and $b$ denote one of their generic agent, respectively.

\paragraph{Benefits from interaction and costs}
Agents benefit from trading/interacting with other agents and, in particular, any agent $a\in A$ receives a gross utility of $p_A$, when trading in her home country $A$, and a possibly higher gross utility $p_B$, when instead exporting to the other country $B$. 
Benefits are assumed to be equal across agents and to be decreasing functions of the current infection prevalence rates $x_A(t), x_B(t) \in [0,1]$.\footnote{These can be interpreted as prices at which trade happens.}
This last assumption reflects the fact that trading becomes riskier as contagion spreads. Formally:
$$
p_A = p_A(x_A(t),x_B(t)), \qquad p_B = p_B(x_A(t),x_B(t)),
$$
for any time $t\in \R$. 

Any generic agent $a \in A$ chooses between two (mutually exclusive) actions, respectively labeled as $A$ and $B$, which are either ``trading in her home country'' or ``exporting to the other country''.\footnote{According to our notation, then, agents $a\in A$ can export to $B$ and, conversely, agents $b\in B$ to $A$.}
However, to export to the other country each agent $a\in A$ has to pay an exporting cost $c_a>0$, which is assumed to be randomly distributed across agents according to a cumulative distribution function $F_A$. 
The cost of trading in the home country, instead, is normalized to 0.
Depending on the chosen action $A$ or $B$, agent $a$'s utility at time $t$ is then given by:
$$
u_a(t) =
\left\{
\begin{array}{ll}
  p_A(t), & \mbox{if trading in } A,  \\
  p_B(t) - c_a, & \mbox{if exporting to } B,
\end{array}
\right.
$$
so that agent $a \in A$ decides to export to $B$ at time $t$ if and only if
$$
p_A < p_B - c_a.
$$
Symmetrically, analogous definitions and notations hold for all agents $b \in B$. The same happens in the rest of this section.


\begin{remark}
Notice that in our formulation agents make decisions only based on the prices $p_A$ and $p_B$ that they are able to observe in the two markets. 
In particular, they are not able to observe neither their status (as susceptible or infected) nor the status of the others.
In the case of the cattle trade mentioned in the motivation section, this assumption is economically justified as follows. Movements are stressful for bovines, which can result in the development of diseases and reduced growth (or even death) of the animals. This also implies that latent diseases can be masked as stress and go undetected. For these reasons, farmers are compelled to report to local health institutions any situation that may be related to a disease -- they would otherwise incur in high fines.
\end{remark}

Since $c_a \sim F_A$, the above expression implies that the fraction of $A$'s agents willing to export to $B$ at time $t$ is given by
\begin{equation}
\begin{aligned}
\label{eq: willing_to_export}
\mathbb P\left\{a\in A: p_A(t) < p_B(t)-c_a\right\} & = \mathbb P\left\{a\in A: c_a < p_B(t)-p_A(t)\right\} \\
& {} = F_A\left\{p_B(t)-p_A(t)\right\},
\end{aligned}
\end{equation}
or, equivalently, that the fraction of $A$'s agents trading in $A$ and \textit{not} exporting is $1-F_A\left\{p_B(t)-p_A(t)\right\}$.

\paragraph{Cross-country meetings and flows of infected individuals}
Let us proceed with the analysis: of the fraction of agents that are exporting from $A$ to $B$, a subfraction of them given by $x_A\cdot F_A\left\{p_B-p_A\right\}$ is of currently infected agents. Consequently, when these exporting and infected agents meet the fraction of those susceptible in $B$ that remain in $B$ for trade, which is $(1-x_B)(1-F_B\{p_A-p_B\})$, this will give rise to an additional source of infected individuals for country $B$:\footnote{For ease of notation, throughout we will write $F_A\{p_B-p_A\} = F_A$ and $F_B\{p_A-p_B\}=F_B$.}
$$
\underbrace{x_A \cdot F_A}_{A\text{'s infected exporting to } B} \cdot \underbrace{(1-x_B)\cdot(1-F_B)}_{B\text{'s susceptible remaining in }B}.
$$
Still another source of infection for $B$ comes from the meetings between $B$'s infected individuals remaining in $B$ with $A$'s susceptible exporting to $B$:
$$
\underbrace{x_B\cdot(1-F_B)}_{B\text{'s infected remaining}} \cdot \underbrace{(1-x_A)\cdot F_A}_{A\text{'s susceptible exporting to }B}.
$$
However, this additional infective activity due to cross-country interactions is somehow compensated with a reduction in the home country. In particular, now $B$'s within-country spreading cannot follow the single-location equation (\ref{SI_cubic}) given in Section \ref{single-locationmodel}: not only because the meetings only happen between $B$'s susceptible and infected agents that are not exporting, but also because we have to subtract the fraction of $B$'s infected agents that are exporting, as an outflow.
$$
\underbrace{\nu_B x_B(1-F_B)(1-x_B)(1-F_B)(x_B-q)}_{\text{meetings among } B\text{'s remaining agents resulting in infections}} - \underbrace{x_B F_B}_{\text{outflow of infected}},
$$
where $\nu_B \in (0,1)$ is the contagiousness parameter for $B$. Analogous reasonings hold symmetrically for $A$.

By putting all these elements together, we can build a system of coupled differential equations ruling the evolution over time of the infection rates in the two countries. The first line of each equation accounts for the possibly reduced within-country epidemic spreading, whereas the second line accounts for the additional inflow of infection due to cross-country interactions just described above:\footnote{For ease of notation, we omit the time $t$. However, it is worth remembering that $F_A$ and $F_B$ depend on $p_A$ and $p_B$ which, in turn, depend on $x_A(t)$ and $x_B(t)$.}
\begin{equation}\label{ODE}
\left\{
\begin{aligned}
\frac{\de}{\de t}x_A = {} & \nu_A \Big[ x_A(1-F_A)(1-x_A)(1-F_A)(x_A-q_A) +{}\\ 
& {} + x_A(1-F_A)(1-x_B)F_B + (1-x_A)(1-F_A)x_B F_B \Big] - x_A F_A \\
\frac{\de}{\de t}x_B = {} & \nu_B \Big[ x_B(1-F_B)(1-x_B)(1-F_B)(x_B-q_B) +{}\\ 
& {} + x_B(1-F_B)(1-x_A)F_A + (1-x_B)(1-F_B)x_A F_A \Big] - x_B F_B,
\end{aligned}
\right.
\end{equation}
where $\nu_A,\nu_B \in (0,1)$ and $q_A,q_B \in (0,1)$ are the contagiousness and quarantine parameters respectively of location $A$ and $B$. The system can be algebraically rearranged as follows:
\begin{equation}
\left\{
\begin{aligned}
\frac{\de}{\de t}x_A = {} & \nu_A (1-F_A) \Big[ x_A(1-x_A)(x_A-q_A)(1-F_A) + (x_A+x_B-2x_Ax_B)F_B \Big] \\
& {} - x_A F_A\\
\frac{\de}{\de t}x_B = {} & \nu_B (1-F_B) \Big[ x_B(1-x_B)(x_B-q_B)(1-F_B) + (x_A+x_B-2x_Ax_B)F_A \Big] \\ 
& {} - x_B F_B.
\end{aligned}
\right.
\end{equation}

\begin{remark}
	If cross-country export is not allowed, i.e. when $F_A = F_B = 0$, then system (\ref{ODE}) is reduced to two uncoupled equations, corresponding to two single-location models of the form of equation (\ref{SI_cubic}), both evolving separately.
\end{remark}

\begin{remark}
In this model, export at any instant only occurs in one direction at time, either from $A$ to $B$ or vice versa. Indeed, suppose that $p_A(t) < p_B(t)$ at a certain time $t\in \R$. Since $F_A$ and $F_B$ are cumulative distributions which are positive only for positive costs, then in such a case $F_B(p_A(t) - p_B(t)) = 0$ while $F_A(p_A(t) - p_B(t)) > 0$. So, there is an outflow of infection in the first equation for $x_A$ and an inflow in the second for $x_B$. However, as the following analysis will show, the infection rates $x_A(t)$ and $x_B(t)$ (as well as $p_A(t),p_B(t)$) are not necessarily monotone functions of time.
\end{remark}

In the following, we will make the assumption that $F_A = 1$ when $x_A = 1$.
This is intuitive: whenever in $A$ the rate of infection is the maximum, i.e. $x_A = 1$, then all $A$'s agents would be facing the minimum home benefit $p_A$ and thus be willing to export, so $F_A = 1$.

\begin{proposition}
\label{prop_system_well_defined}
System \emph{(\ref{SIcubic})} is well defined in the unit square describing any $(x_A,x_B)\in [0,1]^2$.
\end{proposition}
\begin{proof}
See \ref{sec: proofs}.
	
\end{proof}

\subsection{Identical locations, linear utility \& uniform cost}\label{simplestexample}

To keep the analysis tractable, we restrict our model to a linear specification of system (\ref{ODE}):
\begin{itemize}
\item   the two locations $A$ and $B$ are assumed to be identical, from the point of view of the epidemic parameters, so $\nu_A = \nu_B = \nu \in (0,1)$ and $q_A=q_B=q\in (0,1)$;
\item   the agents' exporting costs $c_a>0$, for $a\in A$, are uniformly distributed over the interval $[0,1]$ (analogously for $b \in B$), so that the cumulative distributions are identical and of the form $F_A = F_B = \mathcal U(0,1)$:
    $$
    F_A(c) = F_B(c) =
    \begin{cases}
    0,              &   \text{ for } c \leq 0  \\
    c,              &   \text{ for } c \in [0,1]  \\
    1,              &   \text{ for } c \geq 1.
    \end{cases}
    $$
    In particular, the maximum and minimum cost are respectively 1 and 0.
\item   The gross utilities from trading, $p_A$ and $p_B$, are assumed to depend linearly on the infection rate of the own location:
    $$
    p_A(x_A(t),x_B(t)) := 1 - x_A(t), \qquad p_A(x_A(t),x_B(t)) := 1 - x_B(t).
    $$
    Then, maximum and minimum gross utility attainable are thus normalized to 1 and 0, respectively.
\end{itemize}
With these assumptions in place, equation (\ref{eq: willing_to_export}) becomes:\footnote{Remember that $x_A,x_B \in [0,1]$.}
$$
\begin{aligned}
F_A\{p_B - p_A\} & = F_A\{x_A-x_B\} = 
\begin{cases}
0, 			&	\text{if } x_A-x_B < 0 		\\
x_A-x_B,	&	\text{if } 0\leq x_A-x_B \leq 1	\\
1,			&	\text{if } 1 < x_A-x_B
\end{cases} \\
& = \max\{0,x_A-x_B\},
\end{aligned}
$$
and, analogously, $F_B = \max\{0, x_B-x_A\}$. We can then rewrite system (\ref{ODE}) as follows:
\begin{equation}\label{ODEmax}
\left\{
\begin{aligned}
\frac{\de}{\de t}x_A = {} & \nu (1-\max\{0,x_A-x_B\}) \Big[ x_A(1-x_A)(x_A-q)(1-\max\{0,x_A-x_B\}) \\ 
& {}+(x_A+x_B-2x_Ax_B)\max\{0,x_B-x_A\} \Big] - x_A \max\{0,x_A-x_B\}\\
\frac{\de}{\de t}x_B = {} & \nu (1-\max\{0,x_B-x_A\}) \Big[ x_B(1-x_B)(x_B-q)(1-\max\{0,x_B-x_A\}) \\
& {} + (x_A+x_B-2x_Ax_B)\max\{0,x_A-x_B\} \Big] - x_B \max\{0,x_B-x_A\}.
\end{aligned}
\right.
\end{equation}

In \ref{app_linear} we derive the properties of this system, which can be summarized as follows.
The system is well defined in the unit square $[0,1]^2 \subset \R^2$, which is invariant under its dynamics, and it is symmetric with respect to the diagonal in $\R^2$ (Propositions \ref{proposition_system_well_defined} and \ref{proposition_unit_square_invariant}).
This system has three equilibria (Proposition \ref{prop_app2}):
\begin{itemize}
	\item	$(x_A,x_B)=(0,0)$ and $(1,1)$, which are asymptotically stable states;
	\item	$(x_A,x_B)=(q,q)$, which is an unstable saddle point.
\end{itemize}
What becomes interesting is to study the basins of attractions of the two stable equilibria, and to characterize the basins' border, which we call \emph{separatrix} $\mathcal C$.
As shown in Figure \ref{basinattraction}, depending on the parameters $\nu$ and $q$, as time $t$ passes, the solution enters the unit square either crossing its border along the segment $[q,1]\times \{0\}$ or along $\{1\} \times [0,q]$ and, eventually, converges toward $\{q,q\}$ as $t\rightarrow \infty$ (Proposition \ref{proposition_separatrix}).

\begin{figure}[tbp]
	\centering
	\caption{\textbf{Basins of attraction of the disease-free and fully-endemic equilibria, $(0,0)$ and $(1,1)$, for changing epidemic parameters $\nu,q$.}}
	\begin{minipage}[t]{.49\textwidth}
		\centering\setlength{\captionmargin}{0pt}
		\includegraphics[width=\textwidth]{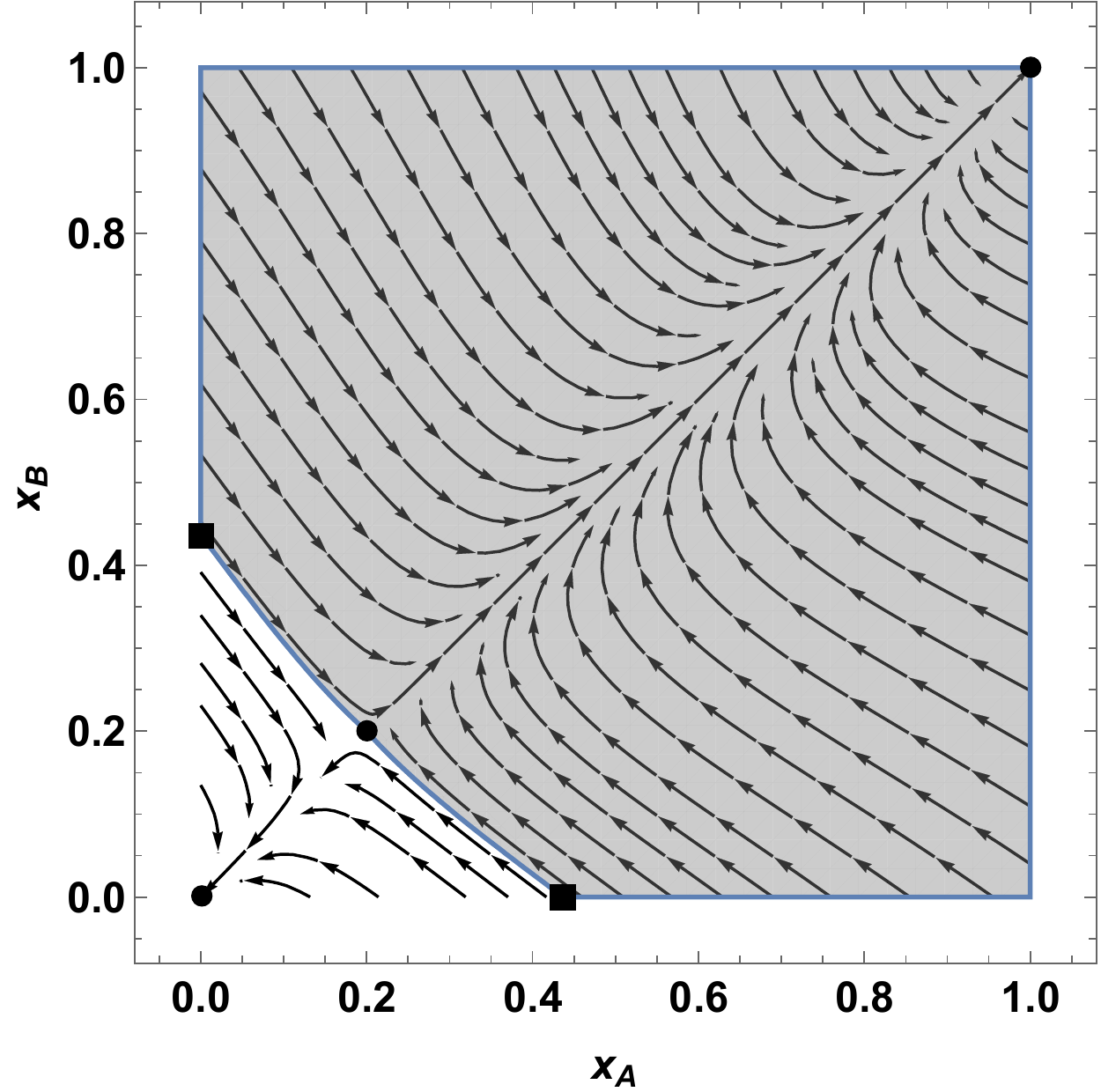}
		\caption*{\small Parameters $\nu=0.7$, $q=0.2$.}
	\end{minipage}
	\hfill
	\begin{minipage}[t]{.49\textwidth}
		\includegraphics[width=\textwidth]{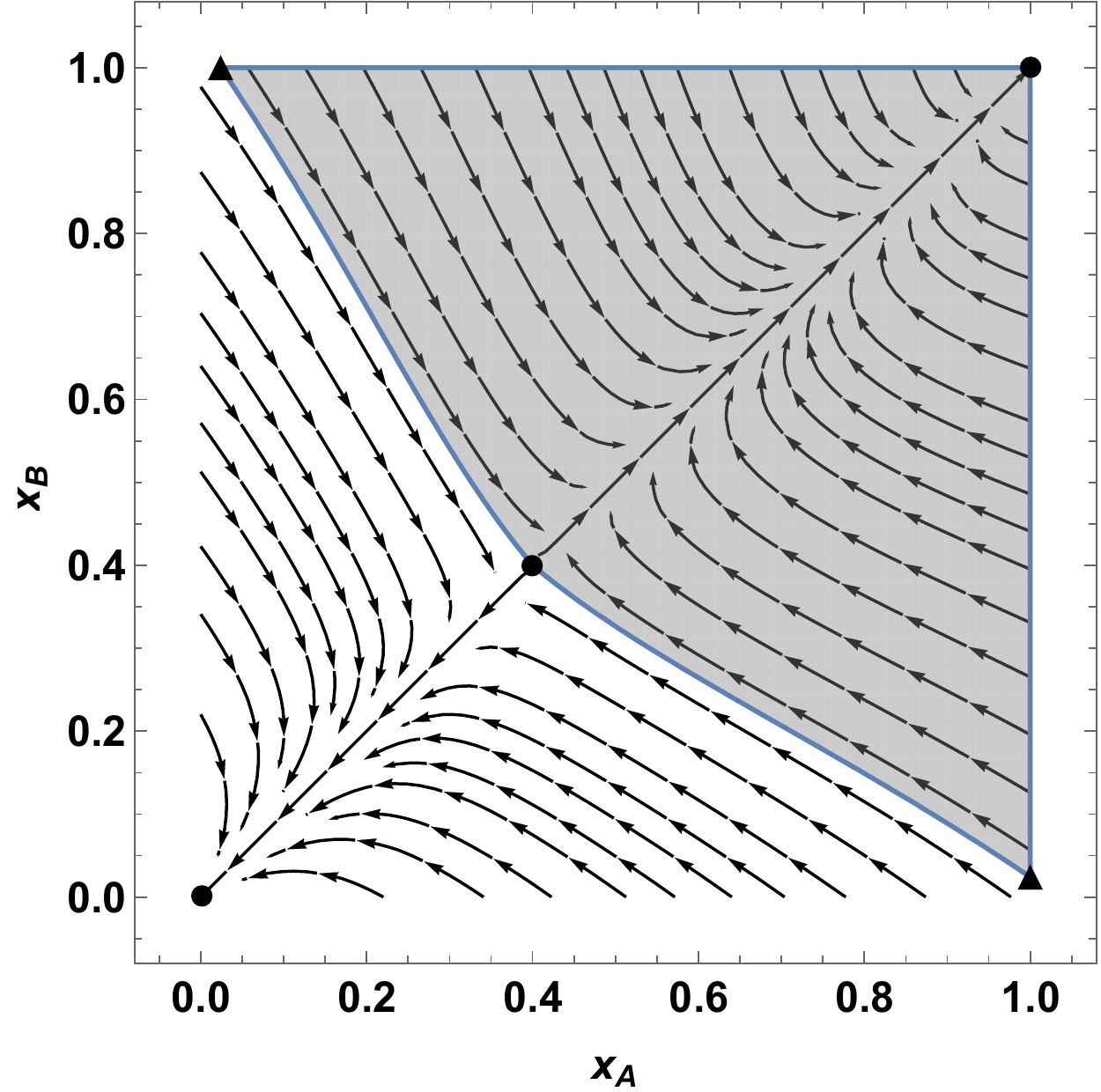}
		\caption*{\small Parameters $\nu=0.7$, $q=0.4$.}
	\end{minipage}
	\hfill
	\begin{minipage}[t]{.49\textwidth}
		\centering\setlength{\captionmargin}{0pt}
		\includegraphics[width=\textwidth]{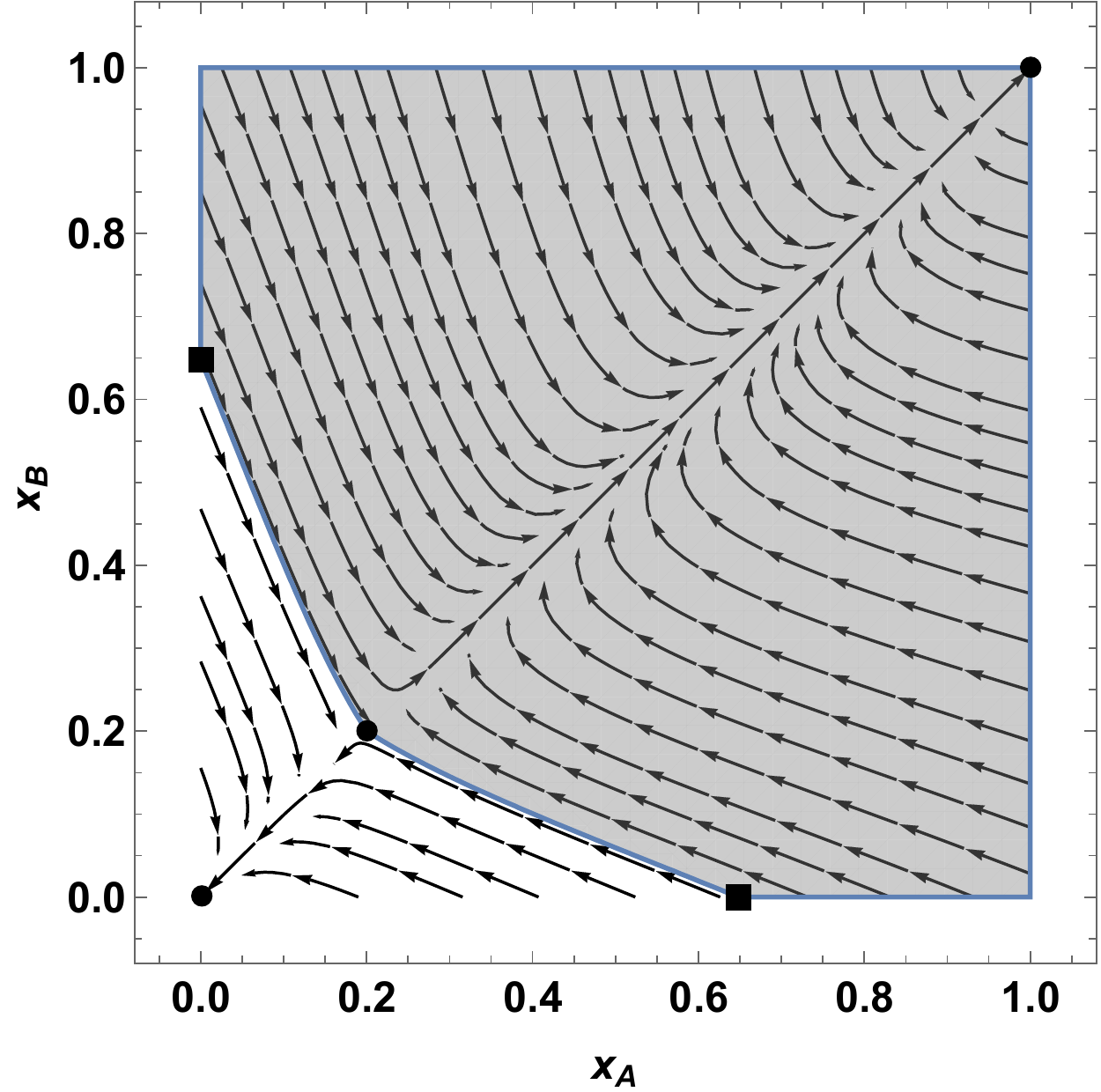}
		\caption*{\small Parameters $\nu=0.4$, $q=0.2$.}
	\end{minipage}
	\hfill
	\begin{minipage}[t]{.49\textwidth}
		\includegraphics[width=\textwidth]{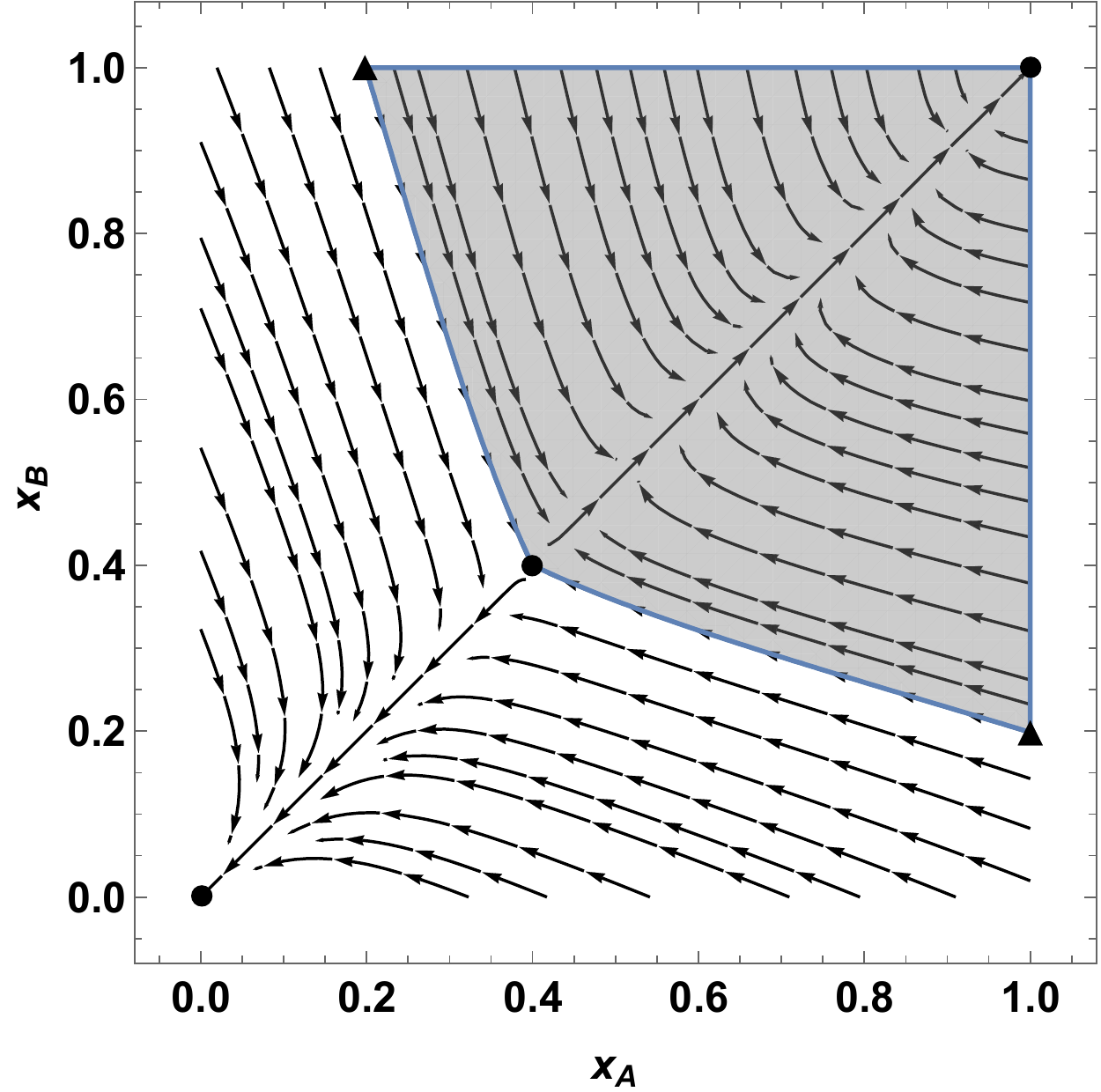}
		\caption*{\small Parameters $\nu=0.4$, $q=0.4$.}
	\end{minipage}
	\caption*{\small Simulations in \textsc{Mathematica}\textsuperscript{\textregistered} to plot the vector field defining system (\ref{ODEmax}) and the basins of attraction of the two asymptotically stable states $(x_A,x_B)=(0,0)$ and $(1,1)$ (respectively colored in white and shaded gray). The arrows depict the vector field defining system (\ref{ODEmax}) in each point $(x_A,x_B) \in [0,1]^2$ and confirm that the unit square is invariant and that the same is true for the diagonal and for the super-diagonal and sub-diagonal ``triangles''. Moreover, from the saddle point $(q,q)$ one can identify the separatrix curves, the unstable separatrix coinciding with the diagonal, while the stable one, i.e. $\mathcal C$, constitutes part of the border of the basins of attraction, thus separating them. Lastly, notice that as quarantine $q$ increases, the system exhibits a larger and larger basin of attraction of the disease-free equilibrium $(0,0)$, which is intuitively due to the fact that it is easier to recover from infection. The squared dots are $(\eta,0)$ and $(0,\eta)$, i.e. the intersection points of the separatrix $\mathcal C$ with the horizontal and vertical axis mentioned in Proposition \ref{proposition_separatrix}. Analogously, the triangular dots are $(1,\zeta)$ and $(\zeta,1)$.}
\label{basinattraction}
\end{figure}  

In \ref{app_comparative} we show that, although the separatrix $\mathcal C$ cannot be described analytically, it can be very well approximated linearly.
In \ref{app_comparative} we also show the good accuracy of this approximation, done by mean of an extensive grid of numerical simulations.
This will allow us to make a comparative statics analysis that will be then used in the following Section \ref{sec_discussion}.
More precisely, in Proposition \ref {lemmalinearseparatrix} and Lemma \ref{lemma_point_intersection} we give an explicit approximation $\widetilde{\mathcal C}$ of $\mathcal C$ and of its point of intersection with the boundaries of the unit square $[0,1]^2$.


Figure \ref{linearseparatrix} shows similarities and differences between $\mathcal C$ and $\widetilde{\mathcal C}$.
Depending on the parameters $\nu$ and $q$, the area under the curve $\widetilde{\mathcal C}$ is either a trapezoid or a triangle and is easily computed analytically. By considering this area as an approximation of the area under the curve $\mathcal C$, which is instead impossible to compute analytically. This area obtained with this linear approximation will also be used for a comparative statics analysis in Section \ref{sec_discussion}. The results are also shown in Figure \ref{approximatedarea}.

\begin{figure}[tbp]
	\centering
	\caption{\textbf{Separatrix $\mathcal C$ and comparison with its linear approximation $\widetilde{\mathcal C}$}}
	\begin{minipage}[t]{.49\textwidth}
		\centering\setlength{\captionmargin}{0pt}
		\includegraphics[width=\textwidth]{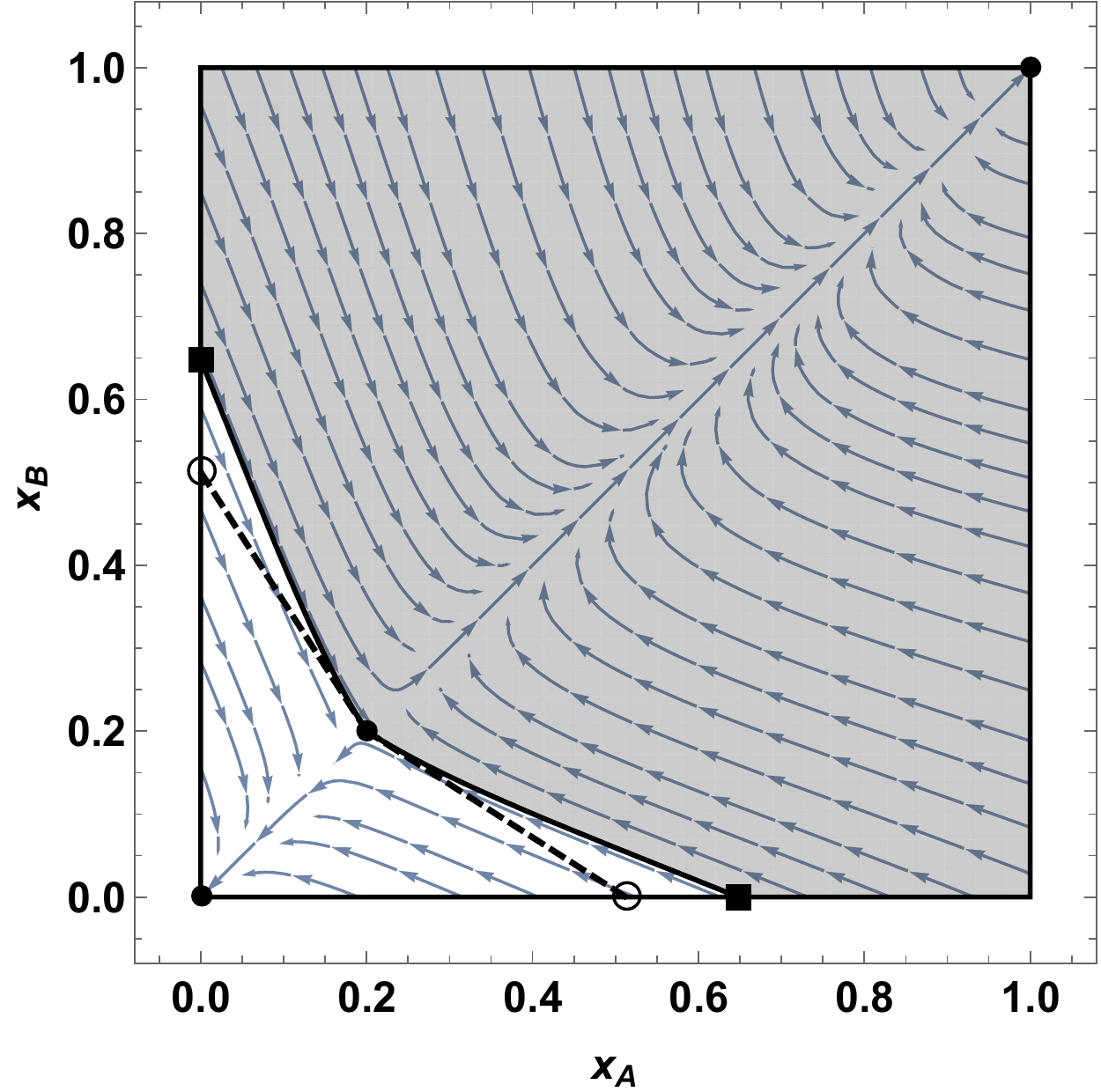}
		\caption*{\small Parameters $\nu=0.4$, $q=0.2$.}
	\end{minipage}
	\hfill
	\begin{minipage}[t]{.49\textwidth}
		\includegraphics[width=\textwidth]{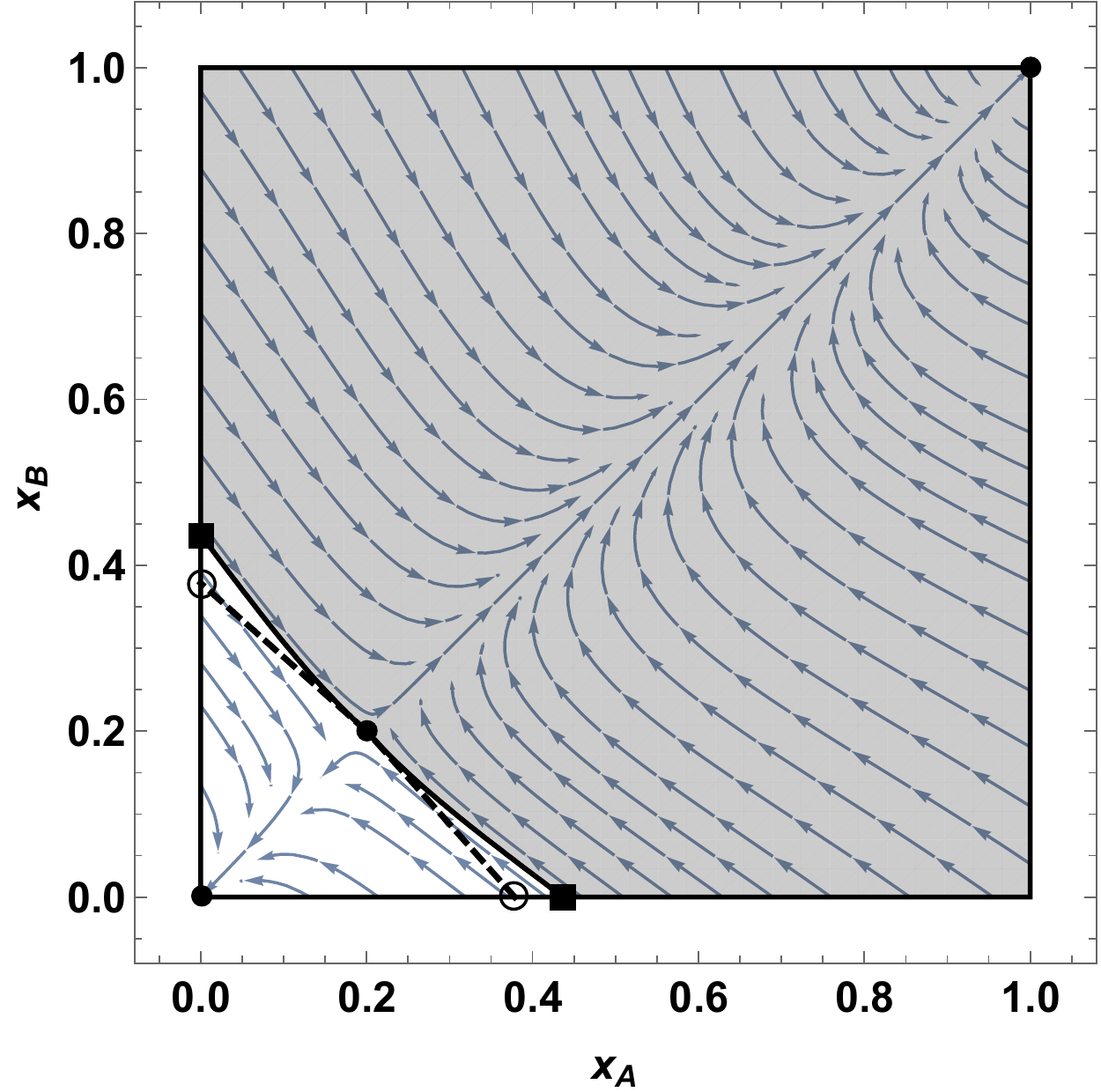}
		\caption*{\small Parameters $\nu=0.7$, $q=0.2$.}
	\end{minipage}
	\caption*{\small Plot of the linearized separatrix $\widetilde{\mathcal C}$ (dashed straight line) and comparison with the actual separatrix $\mathcal C$ (continuous black curve). $\widetilde{\mathcal C}$ is a first-order approximation of $\mathcal C$ in a neighborhood of the saddle $(q,q)$.}
	\label{linearseparatrix}
\end{figure}

\begin{figure}[tbp]
	\centering
	\caption{\textbf{Approximated area of the triangle/trapezoid}}
	\begin{minipage}[t]{.49\textwidth}
		\centering\setlength{\captionmargin}{0pt}
		\includegraphics[width=\textwidth]{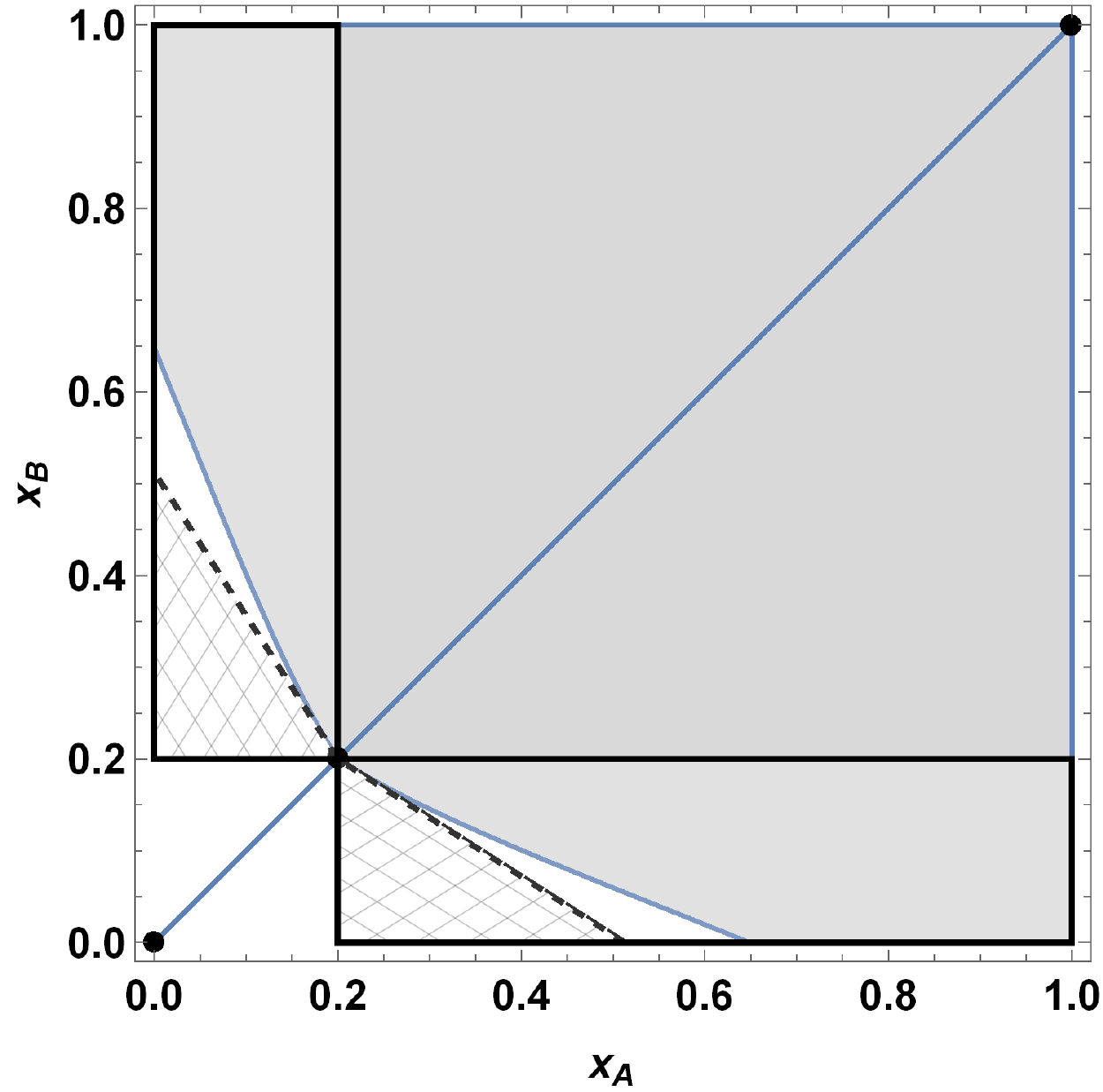}
		\caption*{\small Parameters $\nu=0.4$, $q=0.2$.}
	\end{minipage}
	\hfill
	\begin{minipage}[t]{.49\textwidth}
		\includegraphics[width=\textwidth]{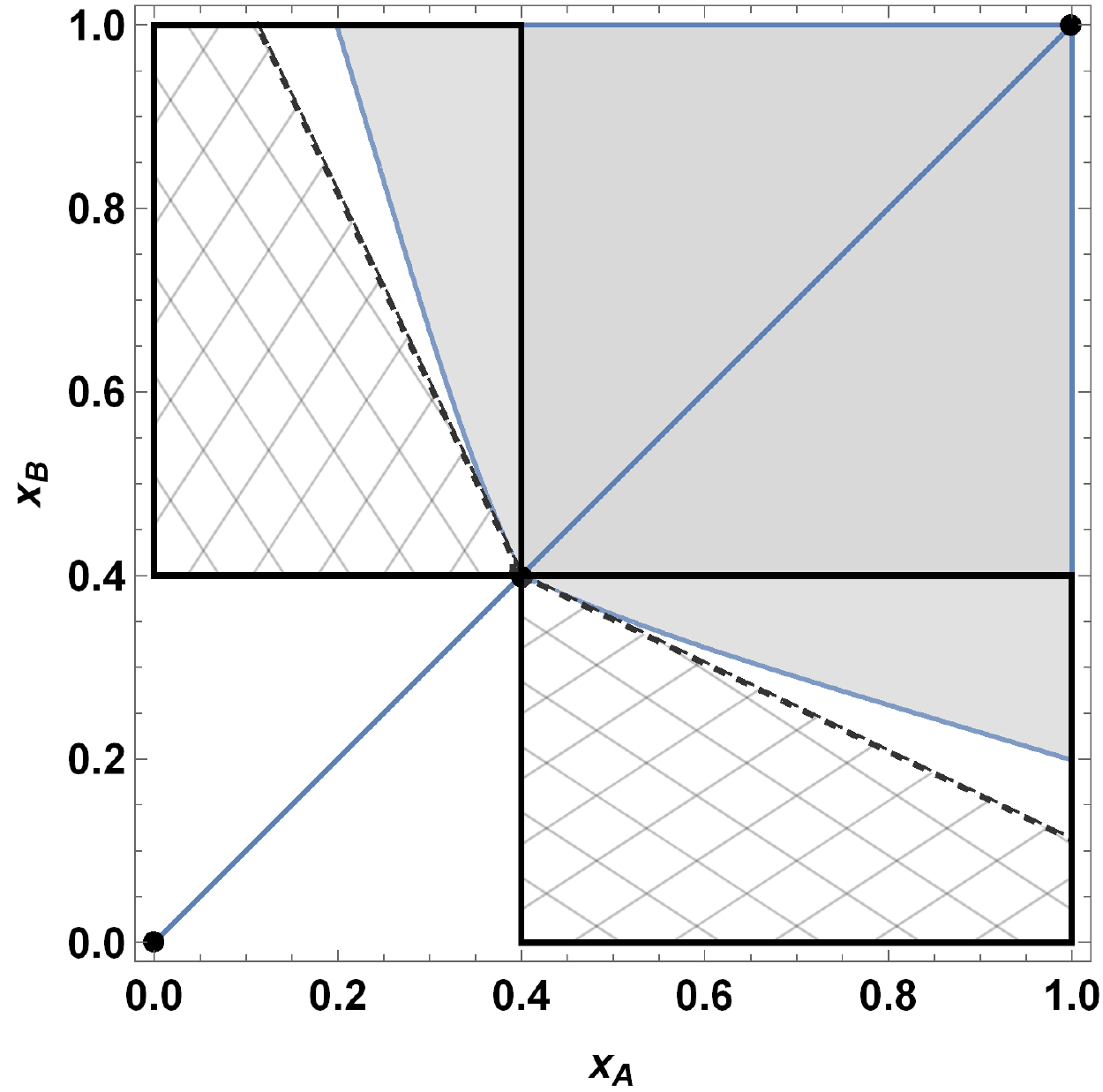}
		\caption*{\small Parameters $\nu=0.4$, $q=0.4$.}
	\end{minipage}
	\caption*{\small Approximated areas of the triangles/trapezoids (grid-shaded areas), defined by the linear approximation $\widetilde{\mathcal C}$ (dashed lines), in the rectangles of interest.}
	\label{approximatedarea}
\end{figure}

\clearpage

\section{Comparative statics with respect to exogenous shocks}
\label{sec_discussion}

We now focus our attention on the conclusions that can be drawn from the analysis of system (\ref{ODEmax}) performed in Section \ref{simplestexample}, \ref{app_linear} and \ref{app_comparative}. We will compare the two following situations:
\begin{itemize}
	\item	first, no cross-country trade between the two locations is allowed, i.e. they are considered separated and autarkic;
	\item	second, cross-country trading is instead allowed, as described in the previous section.
\end{itemize}
By comparing these two situations we are thus able to analyze the effects of a very ``stylized globalization'' (the second case) on the systemic resistance to potential shocks in the infection rates. Depending on the ``intensity'' and ``dimensionality'' of the shock, being ``autarkic'' or ``globalized'' may or may not be advantageous. 

In particular, \textit{small shocks} are better absorbed by an interconnected system, independently of their dimensionality: intuitively, the shock is more easily diluted in a larger system. On the contrary, somehow surprisingly, 
\textit{large shocks} may or may not have worse consequences when the locations are interconnected, depending on the amount of resources dedicated to recovering (formalized by the parameter $q$). 

\subsection{The case of autarky}

Let us consider two autarkic locations, where no trade is possible between them and where each location is subjected to a disease-spread dynamic described by the single-location model of Section \ref{single-locationmodel}. The evolution over time of the two infection rates $x_A(t)$ and $x_B(t)$ of these two locations $A$ and $B$ can be written as a system of two (uncoupled) differential equations:\footnote{For a reasonable comparison with an analogous globalized 2-location model, here we assume the same symmetric epidemic parameters $\nu_A = \nu_B = \nu$ and $q_A = q_B = q$.}
\begin{equation}
\label{ODEautarky}
\left\{
\begin{aligned}
& \frac{\de}{\de t}x_A = \nu x_A (1-x_A)(x_A-q) \\
&  \frac{\de}{\de t}x_B = \nu x_B (1-x_B)(x_B - q).
\end{aligned}
\right.
\end{equation}
The dynamics and results are shown in Figure \ref{comparisonbetweenlocations} (left) and summarized in the following proposition.

\begin{proposition}
\label{proposition_autarky}
Given two autarkic locations $A$ and $B$, system \emph{(\ref{ODEautarky})} has the following properties:
\begin{itemize}
\item   it is symmetric with respect to the diagonal, which is then an invariant set. The super-diagonal and sub-diagonal sets in $\R^2$, $\{(x_A,x_B)\in \R^2:x_A<x_B\}$ and $\{(x_A,x_B)\in \R^2:x_A>x_B\}$, are also invariant;
\item	the unit square $[0,1]^2$ is invariant;
\item	the critical points where $\frac{\de x_A}{\de t} = \frac{\de x_B}{\de t} = 0$ are:
	\begin{itemize}
		\item	$(0,q)$, $(1,q)$, $(q,0)$, $(q,1)$, which all are (unstable) saddle points;
		\item 	$(q,q)$, which is an unstable point (source);
		\item	$(0,0)$, $(0,1)$, $(1,0)$, $(1,1)$, which all are asymptotically stable equilibria.
	\end{itemize}
\end{itemize}
Moreover, the separatrix curves of the saddle points are the lines $x_A = 0$, $x_B = 0$, $x_A = q$, $x_B = q$, $x_A = 1$ and $x_B = 1$, and this also makes possible characterizing the basins of attraction of the stable points in $[0,1]^2$:\footnote{Of course, we only consider their intersection with the unit square, which is the sets in which the fraction of infected $x_A$ and $x_B$ make sense.}
\begin{itemize}
	\item	$[0,q)^2$ is the basin of attraction of $(0,0)$;
	\item	$[0,q)\times (q,1]$ is the basin of attraction of $(0,1)$;
	\item	$(q,1]\times [0,q)$ is the basin of attraction of $(1,0)$;
	\item	$(q,1]^2$ is the basin of attraction of $(1,1)$.
\end{itemize}
\end{proposition}
\begin{proof}
	See \ref{sec: proofs}.
\end{proof}

As shown in Figure \ref{comparisonbetweenlocations} (left), the points $(1,0)$ and $(0,1)$ play a peculiar role: they represent a situation in which only one of the two locations is fully infected, while the other is disease free. In case of autarky, this may happen when the initial point of infection at time $t=0$ belongs to $[0,q)\times (q,1]$ or $(q,1]\times [0,q)$, which will cause the dynamics to convergence toward $(0,1)$ or $(1,0)$, respectively. Figure \ref{comparisonbetweenlocations} (right) also shows that this cannot be the case when the two locations are connected and globalized. This feature will be a key ingredient in the next section about shock analysis.

\begin{figure}[tbp]
\caption{\textbf{Comparison between autarkic and interconnected locations}}
\begin{minipage}{0.49\textwidth}
\includegraphics[width=\textwidth]{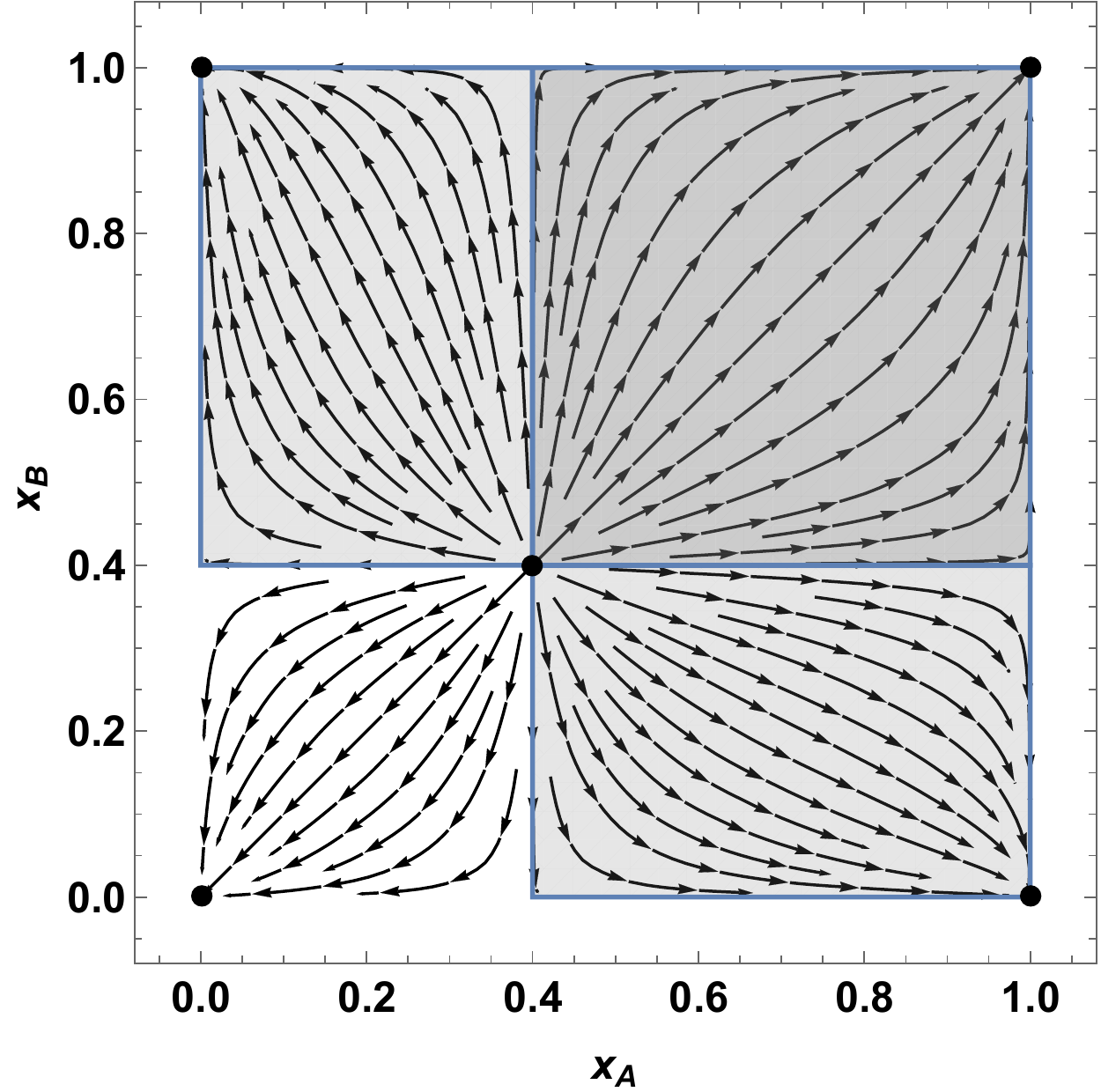}
\caption*{\small Autarky, parameters $\nu=0.7$, $q=0.4$.}
\end{minipage}
\hfill
\begin{minipage}{0.49\textwidth}
\includegraphics[width=\textwidth]{field_basin_connected74.pdf}
\caption*{\small Globalized, parameters $\nu=0.7$, $q=0.4$.}
\end{minipage}
\caption*{\small Dynamics of disease spreading, in case of two autarkic locations (left) and of two globalized locations (right), with same epidemic parameters. In both, the basin of attraction of $(1,1)$ is colored in dark gray, while the basin of $(0,0)$ is left in white. Only the autarkic case (left) exhibits two partially-endemic asymptotically stable states, $(0,1)$ and $(1,0)$, whose basins of attraction are colored in light gray.}
\label{comparisonbetweenlocations}
\end{figure}

\begin{figure}[tb!]
\caption{\textbf{Systemic resistance to small vs. large shocks}}
\begin{minipage}{0.47\textwidth}
\includegraphics[width=\textwidth]{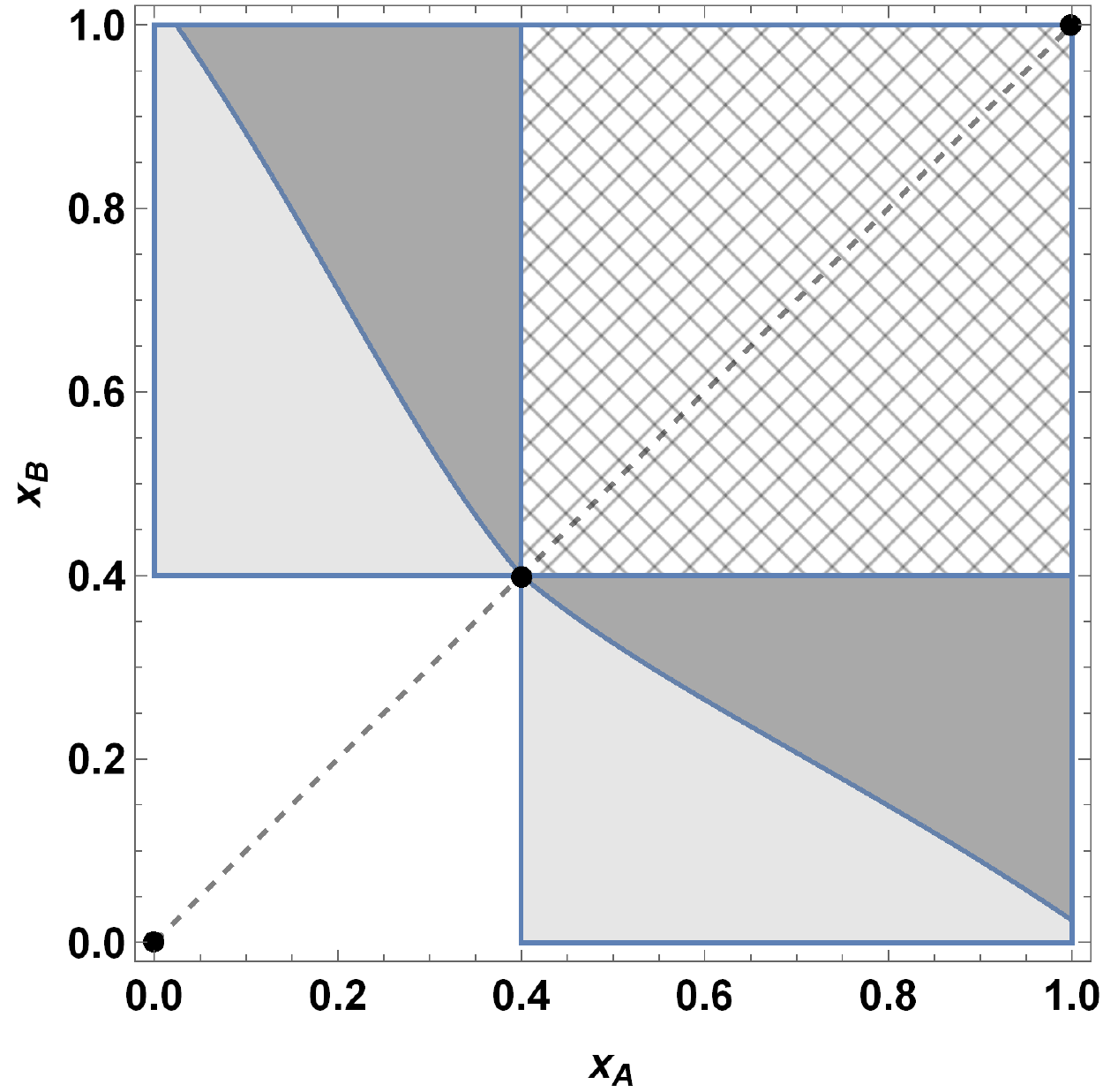}
\end{minipage}
\hfill
\begin{minipage}{0.51\textwidth}
\caption*{\small Comparison by juxtaposition of the areas obtained in Figure \ref{comparisonbetweenlocations}, with same parameters $\nu=0.7$, $q=0.4$. The white area, $[0,q)^2$, and the grid-shaded area, $(q,1]^2$, depict where a hitting shock would produce the same outcome, independently of locations being autarkic or globalized. Light-gray areas measure where shocks result in a partial endemic state, in case of autarkic locations, or where they are instead totally recovered, in case of globalized locations. On the contrary, dark-gray areas are those where shocks result fully infected system, if globalized, whereas only partial infection, if autarkic.}
\label{comparisonbetweenlocations2}
\end{minipage}
\end{figure}

\subsection{Shock analysis}

In line with what done in Section \ref{sec_1island}, \emph{shocks} are assumed to be uniform at random over the unit square $[0,1]^2$ and are represented by a vector of initial conditions:
$$
\mathbf s = (s_A,s_B) := \left(x_A(0),x_B(0)\right).
$$
Figures \ref{comparisonbetweenlocations} and \ref{comparisonbetweenlocations2} show the comparison of the basin of attraction of the point $(1,1)$, when the locations are considered autarkic or connected. Due to the shape of these basins of attraction, we obtain the following results.\footnote{%
We assume a uniform distribution of shocks just to simplify the exposition.
What is important for our analysis, is just that the support of our random shocks is the unit square $[0,1]^2$, so that we can compare regions of this support in the two regimes of \emph{atarky} and \emph{globalization}.
With uniform shocks areas in the support region are proportional to probabilities, but all following results could be adapted also to any other distribution of shocks, remembering that in the more general case areas should be \emph{translated} into probabilities.}

Given a shock $\mathbf s = (s_A,s_B)$, if it is large enough in both components or small enough in both components, then the resulting outcome is the same for an autarkic system and for a globalized system. 
In particular:\footnote{Notice that assuming that the shock distribution is uniform or continuous, thus atom-less, guarantees that the probability that one component hits $0$, $q$ or $1$ is zero.}
\begin{itemize}
	\item	if $s_A < q$ and $s_B < q$, then both the autarkic system and the globalized system will be able to fully recover (white areas in Figure \ref{comparisonbetweenlocations2});
	\item	if $s_A > q$ and $s_B > q$, then both systems will converge to a fully infected endemic state (grid-shaded areas in Figure \ref{comparisonbetweenlocations2}).
\end{itemize}

On the contrary, the outcome resulting from a shock hitting mainly one location is completely different when the two locations are autarkic or connected. Indeed, consider an ``almost'' 1-dimensional shock $\mathbf s$ targeting mainly location $A$, that is\footnote{Symmetrically, the argument is the same for shocks mainly concentrated in $B$.}
$$
\mathbf s = (s_A,\varepsilon), \quad \text{with} \quad \varepsilon < q < s_A.
$$
In the autarkic case, the dynamics will converge to a partial epidemic equilibrium: Proposition \ref{proposition_autarky} and Figure \ref{comparisonbetweenlocations} (left) show that $A$ would converge to fully infection while, independently, $B$ would recover.

Instead, what happens when $A$ and $B$ are connected while facing the same shock $\mathbf s = (s_A,\varepsilon)$ as before? Two different situations may arise:
\begin{itemize}
\item   if $s_A$ is large enough and such that $\mathbf s = (s_A,\varepsilon)$ belongs $(1,1)$'s basin of attraction (dark-gray area in Figure \ref{comparisonbetweenlocations2}) and then the globalized system will end up being fully infected;
\item   if, instead, $s_A$ is still greater than $q$ but not large enough, then $(s_A,\varepsilon)$ belongs to $(0,0)$'s basin (light-gray in Figure \ref{comparisonbetweenlocations2}) and so the globalized system will manage to recover from this shock.
\end{itemize}

This analysis shows that the 2-location autarkic system and the 2-location globalized system react very differently in response to large 1-dimensional shocks.

The dark-gray areas in Figure \ref{comparisonbetweenlocations2} are constituted by all those possible shocks that cause the infection to spread to both locations, when they are connected, or to just one location, when autarkic. Assuming a uniform shock distribution, this area, then, exactly measures the weakness of the system with respect to this kind of mainly 1-dimensional shocks and, in addition, it also captures the advantage of an autarkic system over a globalized one. 

Analogously, but in an opposite way, the light-gray areas in Figure \ref{comparisonbetweenlocations2} capture the advantage of a globalized system over an autarkic one: shocks belonging to these regions are recovered by a connected system, whereas they result in a partial epidemic equilibrium in the autarkic case.\footnote{These results resemble those obtained in the context of financial networks, where agents (e.g. banks) are exposed via financial dependence to others' default and the goal is to understand how shocks spread in a financial network. As argued in \cite{acemoglu2015systemicrisk}: ``as long as the magnitude of negative shocks affecting financial institutions are sufficiently small, a more densely connected financial network [...] enhances financial stability. However, beyond a certain point, dense interconnections serve as a mechanism for the propagation of shocks, leading to a more fragile financial system.'' The same kind of results are achieved in \cite{cabrales2017risk}.}

\begin{figure}[tb]
\caption{\textbf{Comparative statics}}
\begin{minipage}{0.32\textwidth}
\includegraphics[width=\textwidth]{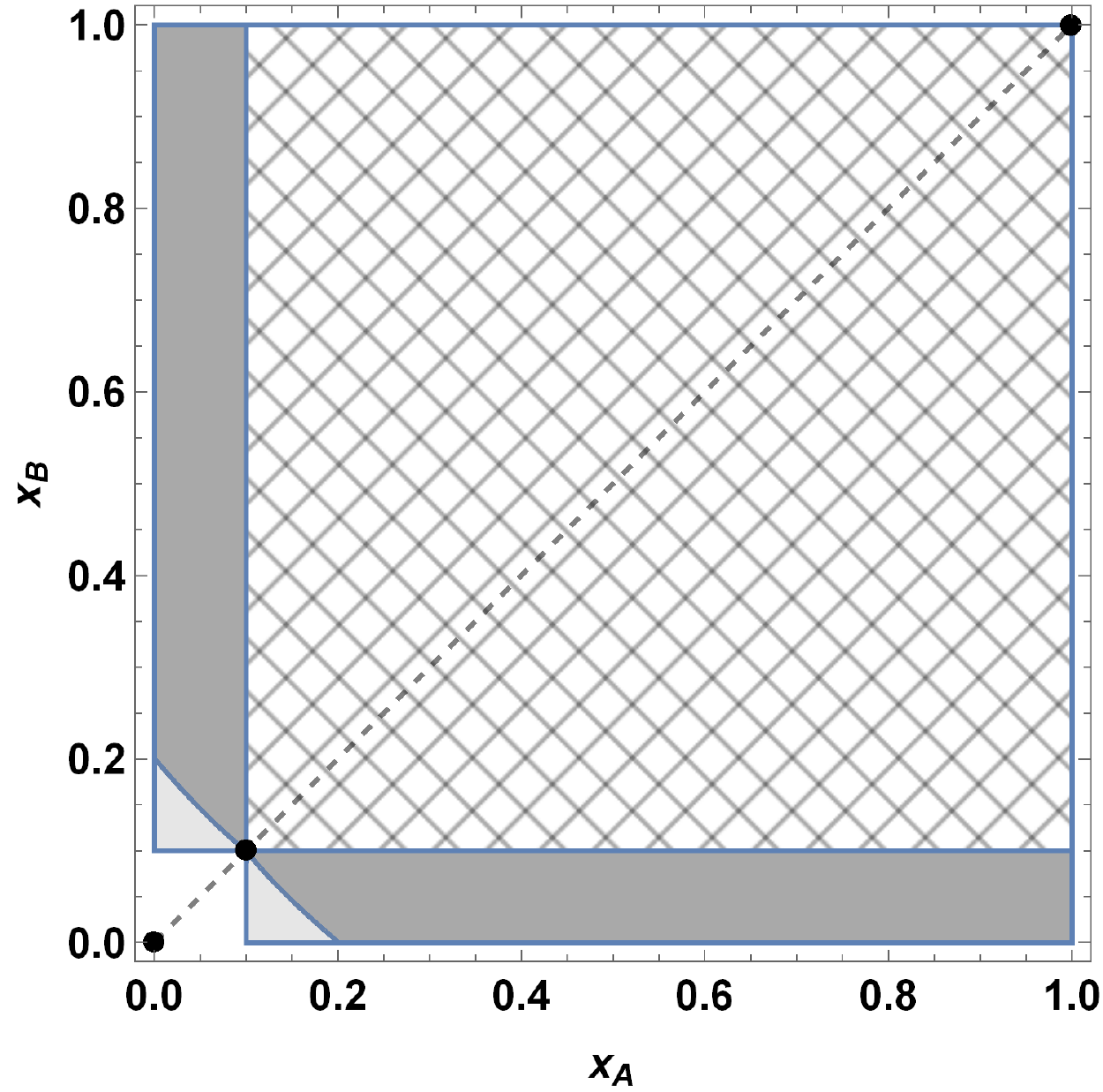}
\caption*{\small $q=0.1$}
\end{minipage}
\hfill
\begin{minipage}{0.32\textwidth}
\includegraphics[width=\textwidth]{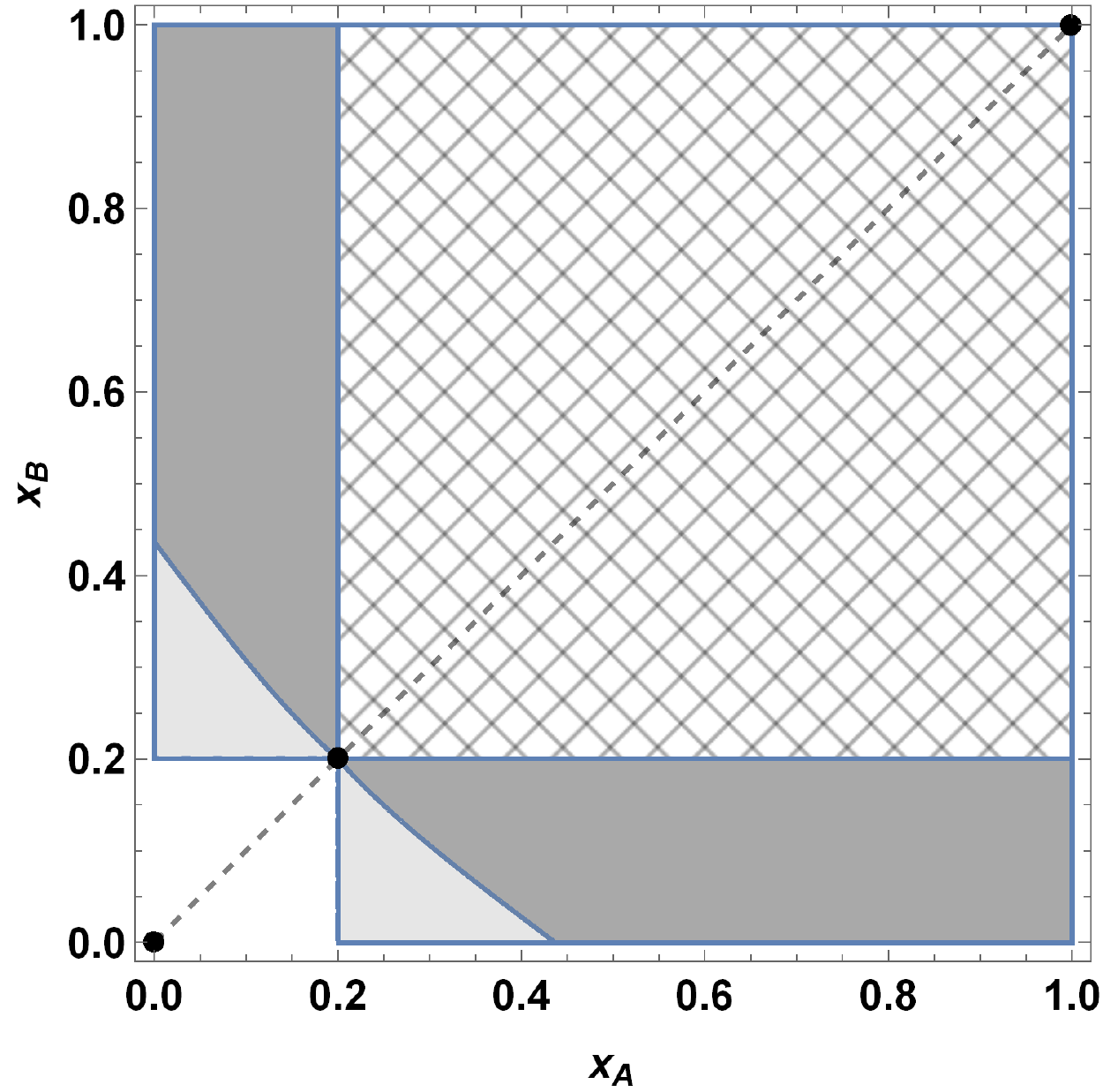}
\caption*{\small $q=0.2$}
\end{minipage}
\hfill
\begin{minipage}{0.32\textwidth}
\includegraphics[width=\textwidth]{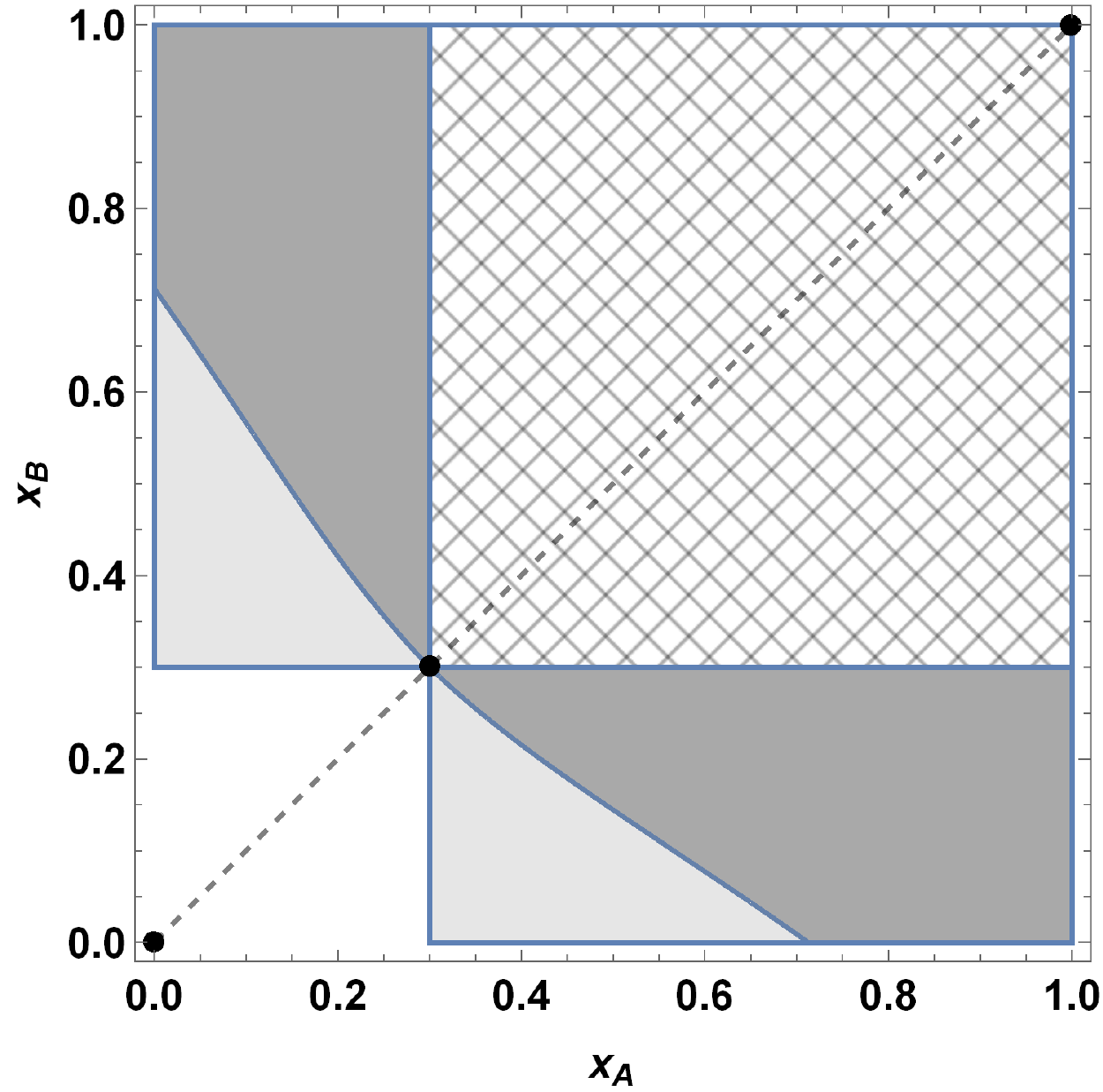}
\caption*{\small $q=0.3$}
\end{minipage}
\begin{minipage}{0.32\textwidth}
\includegraphics[width=\textwidth]{comparison74.pdf}
\caption*{\small $q=0.4$}
\end{minipage}
\hfill
\begin{minipage}{0.32\textwidth}
\includegraphics[width=\textwidth]{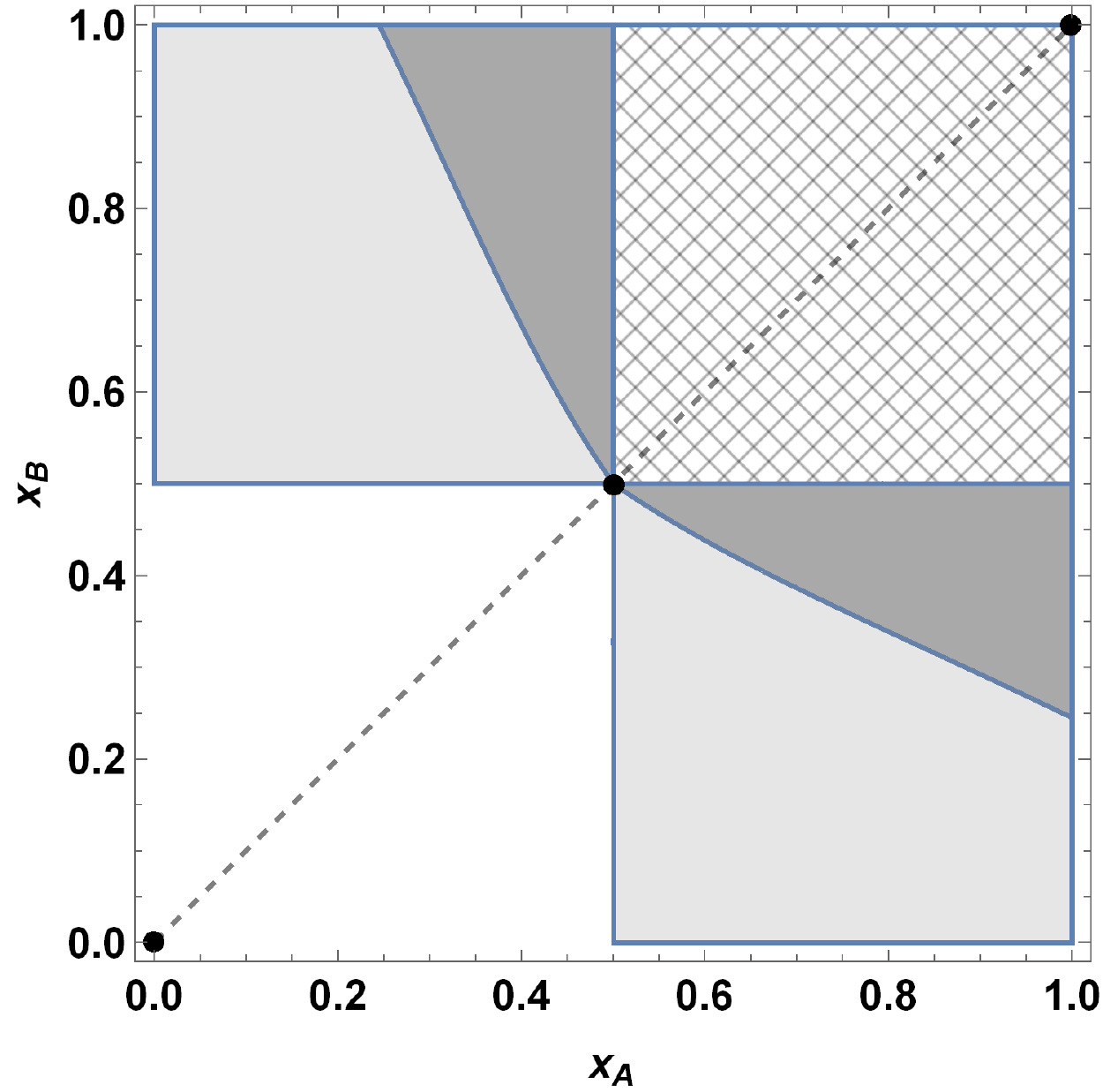}
\caption*{\small $q=0.5$}
\end{minipage}
\hfill
\begin{minipage}{0.32\textwidth}
\includegraphics[width=\textwidth]{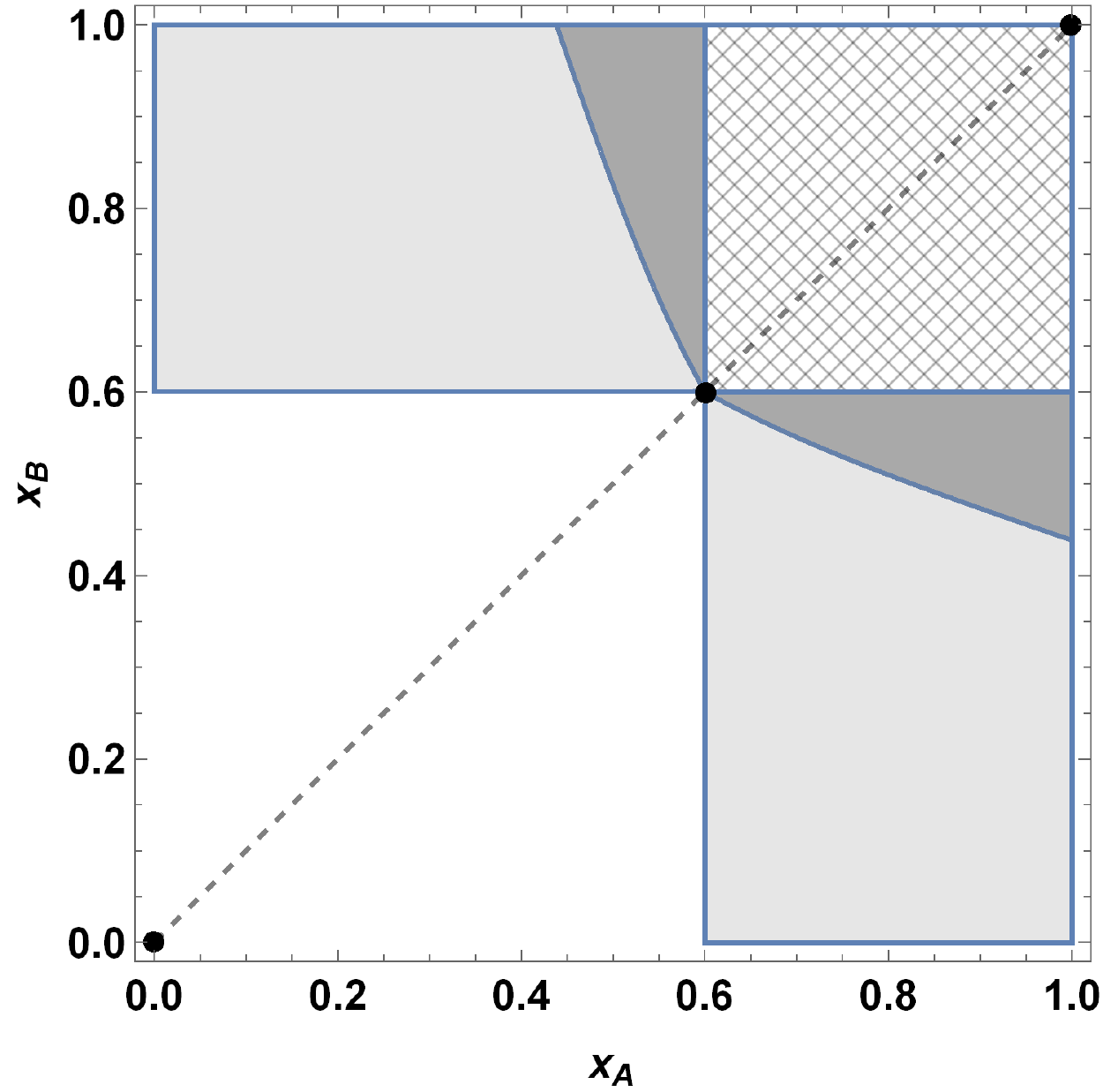}
\caption*{\small $q=0.6$}
\end{minipage}
\caption*{\small Regions of interest when contagiousness $\nu = 0.7$ is kept fixed, while quarantine $q$, the policy parameter, increases in $(0,1)$.}
\label{comparativestatics}
\end{figure}

\subsection{Systemic resistance \& policy}

Understanding the relationship between the dark-gray areas and the light-gray areas in Figure \ref{comparisonbetweenlocations2} becomes necessary, because it gives an indication of the relative (dis)advantage of an autarkic system over a globalized system for systemic resistance. 
Figure \ref{comparativestatics} also shows that this advantage changes as the recovery parameter $q$ varies: this turns out to be crucial for policy making.

One way to address this issue is by analyzing the separatrix curve $\mathcal C$ of the saddle $(q,q)$, because it separates the basins of attraction of the regions of interest. Unfortunately, apart from Proposition \ref{proposition_separatrix}, which relies on the ``local'' information provided by the eigenvector of the linearized system in the neighborhood of the saddle point $(q,q)$ and on the monotonicity of the components of the vector field defining system (\ref{ODEmax}) in specific areas, we have to rely on approximated results, due to the impossibility of explicitly describing $\mathcal C$ analytically.

Specifically, we first numerically approximate the intersection points between the separatrix $\mathcal C$ and the boundaries of the unit square $[0,1]^2$ and, then, numerically measure the gray areas and determine their relative ratio, which, as already observed, is key to understanding whether a globalized system is shock-resistance superior to an autarkic system, given the same parameters $q$ and $\nu$.

\paragraph{(Numerical) comparative statics}
Let us first deal with the (numerical) computation of the intersection point between the separatrix $\mathcal C$ and the border of the unit square below the diagonal, i.e. the segments $[0,1]\times\{0\}$ and $\{1\}\times[0,1]$. \footnote{By symmetry with respect to the diagonal, the same analysis holds also for the border of the unit square above the diagonal.}
Depending on whether $\mathcal C$ intersects the former or latter segment, we follow the notation used in Proposition \ref{proposition_separatrix} and Figure \ref{basinattraction} respectively denote this point with $(x_A,x_B) = (\eta(q,\nu),0)$ or $(1,\zeta(q,\nu))$.

This analysis is shown in Figure \ref{intersectionaxisseparatrix}:
\begin{itemize}
	\item	holding fixed $\nu\in (0,1)$, whenever $\mathcal C$ crosses the segment $[q,1]\times \{0\}$ in the point $(\eta(q,\nu),0)$, then $q\mapsto \eta(q,\nu)$ is increasing in $q$ and spans from $0$ to $1$. Moreover, $\eta(q,\nu) > q$;
	\item	analogously, when $q$ exceeds a certain threshold\footnote{Threshold that corresponds to $q=0.39$ in Figure \ref{intersectionaxisseparatrix}.}, then $\mathcal C$ crosses the segment $\{1\}\times [0,q]$ in the point $(1,\zeta(q,\nu))$; moreover, $q \mapsto \zeta(q,\nu)$ is increasing, going from $0$ to $1$ and always satisfying $\zeta(q,\nu) < q$.
\end{itemize}

Let us now turn to the relative advantage/disadvantage of an autarkic system over a globalized system, especially when subjected to mainly 1-dimensional shocks.\footnote{Shocks that mainly start from a single location, of the form $\mathbf s = (\varepsilon, s_B)$ or $(s_A,\varepsilon)$, with $\varepsilon \approx 0$.}
We have already observed that the areas in light gray and dark gray of Figures \ref{comparisonbetweenlocations2} and \ref{comparativestatics} measure the extent to which an autarkic system or a globalized system is relatively more or less able to recover from shocks of this kind.

Holding fixed the contagiousness $\nu$, as the recovery parameter $q$ increases, the light-gray areas expand while the dark-gray areas shrink.\footnote{We think of contagiousness as a parameter strictly related to the type of disease considered, so not of interest for policy making.} 
According to our previous interpretation, this means that it becomes more likely that a 1-dimensional shock lead the autarkic system to a partial endemic equilibrium, while a corresponding reduction of the dark-gray areas means that a globalized system becomes more able to recover from shocks.\footnote{Since it corresponds to an expansion of the white recovery area, for a globalized system.}
This, in turn, means that the larger it is the available level of quarantine $q$, the more convenient it becomes to be in a globalized system relative to an autarkic one. In this respect, Figure \ref{comparativestatics} shows how the light-gray and dark-gray areas change, as the quarantine $q$ changes.\footnote{While contagiousness $\nu$ is kept fixed, because we think of it as a disease-related parameter, not subject to policy making.} This analysis is also shown in Figure \ref{ratioorangeyellowareas}, where we plot the percentage of the rectangle $[q,1]\times [0,q]$ which is occupied by the dark-gray area. By using the shock analysis done above, as $q$ increases, we observe that having a connected 2-location system becomes more and more advantageous and resistant overall than an autarkic 2-location system. 

This conclusion directly translates in terms of policy: if the available quarantine level $q$ can be taken large enough, then allowing cross-country import-export is beneficial and preferable for systemic resistance to infection shocks.
On the contrary, two autarkic countries constitute a more resistant system against infection shocks when only a small level of quarantine $q$ is available.

%

\begin{figure}[tb]
	\centering
	\caption{\textbf{Intersection between separatrix $\mathcal C$ and boundaries of $[0,1]^2$}}
	\begin{minipage}{.4\textwidth}
		\centering
		\includegraphics[width=\textwidth]{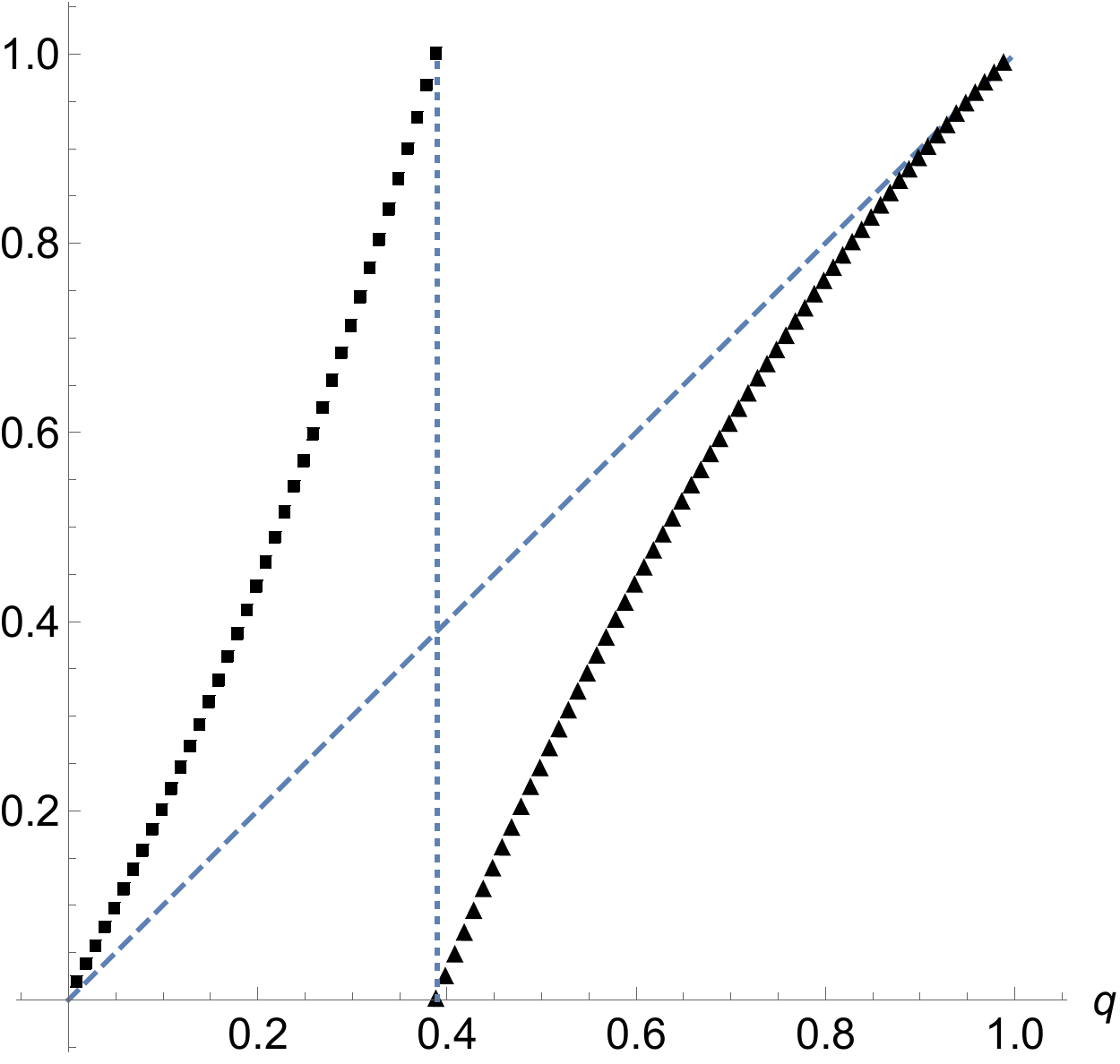}
	\end{minipage}
	\hfill
	\begin{minipage}{.56\textwidth}
		\centering
		\includegraphics[width=\textwidth]{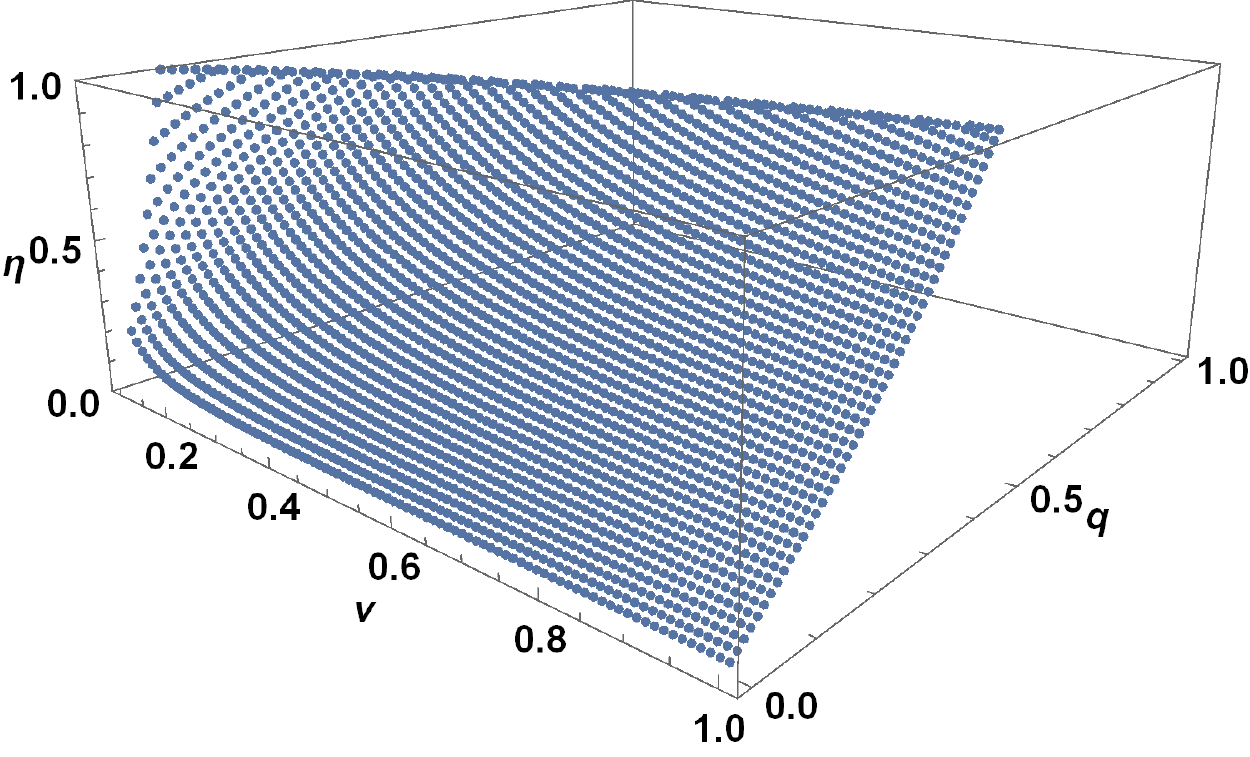}
	\end{minipage}
	\centering
	\caption*{\small On the left, intersection points $q\mapsto \eta(q,0.7)$ (squares) and $q\mapsto \zeta(q,0.7)$ (triangles), with fixed $\nu=0.7$. As $q$ increases, the separatrix $\mathcal C$ first crosses the horizontal segment $[q,1]\times\{0\}$ in $(\eta(q,\nu),0)$, then as $q$ exceeds a certain threshold ($q=0.39$ in this case, signaled by the dotted vertical line), $\mathcal C$ starts crossing the boundary in the vertical segment $\{1\}\times [0,q]$ in the point $(1,\zeta(q,\nu))$. The diagonal (dashed) shows that $\eta > q$ while $\zeta < q$. On the right, intersection $\eta(q,\nu)$ as a function of both parameters $(q,\nu)\in (0,1)^2$. All sections $\eta(\cdot,\nu)$ and $\eta(q,\cdot)$ are increasing.}
	\label{intersectionaxisseparatrix}
\end{figure}

%
%

\begin{figure}[tb]
	\centering
	\caption{\textbf{Ratio between the gray areas}}
	\begin{minipage}{.48\textwidth}
		\centering
		\includegraphics[width=\textwidth]{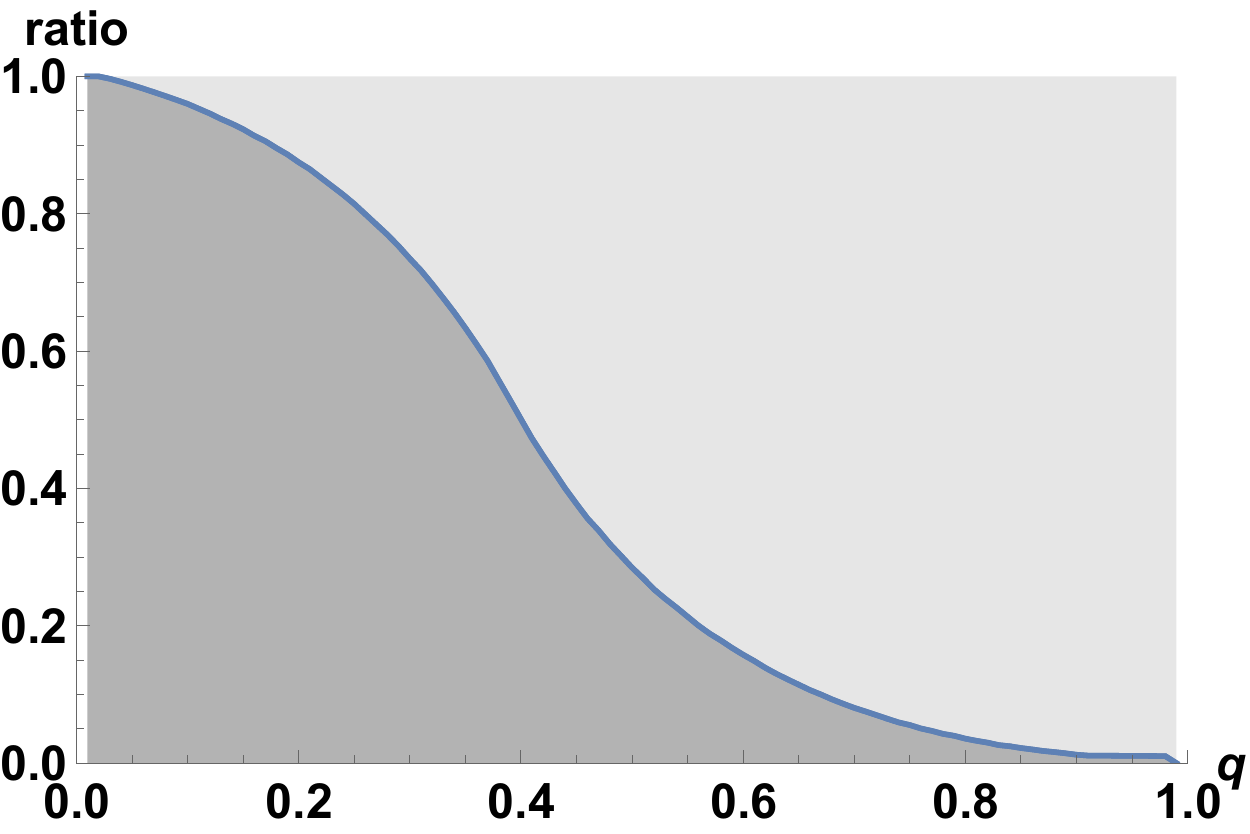}
	\end{minipage}
	\hfill
	\begin{minipage}{.48\textwidth}
		\centering
		\includegraphics[width=\textwidth]{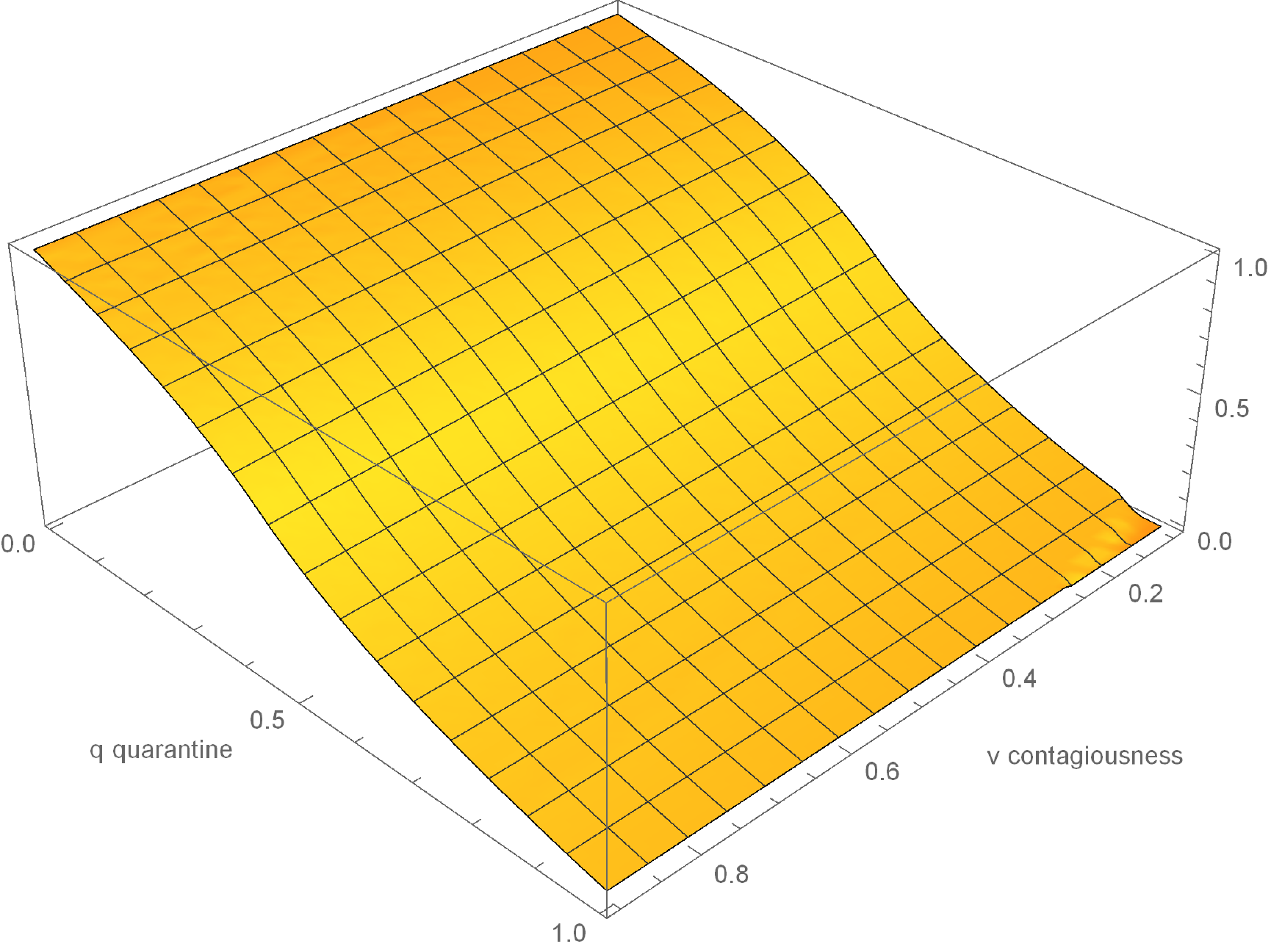}
	\end{minipage}
	\centering
	\caption*{\small On the left, ratio between the dark-gray area and the sum of the dark-gray plus light-gray areas (i.e. $[q,1]\times[0,q] \cup [0,q]\times [q,1]$) numerically obtained as a function of $q\in(0,1)$, holding fix $\nu = 0.7$. On the right, the ratio as a function of both parameters $(q,\nu)\in(0,1)^2$. As the quarantine $q$ increases, a globalized system becomes more and more convenient relative to an autarkic one.}
	\label{ratioorangeyellowareas}
\end{figure}

\section{Conclusions}
\label{sec_conclusion}

Starting from a very simple model of epidemic diffusion among homogeneous agents, we consider the case in which two identical countries are inhabited by such agents. These agents interact and trade with each other in (random) pairs to obtain benefits and, by doing so, they also spread a contagious disease among them, which lowers the attainable gain from trade. As a response to the infection risk, agents can choose to bear (heterogeneous) costs to interact with the agents present in the other country, establishing then a stylized form of cross-country import-export trade. By assuming that both countries have (limited and fixed) resources to intervene against the infection, we are also able to introduce the possibility of recovery, that is, of reducing the infection rate.

Given the epidemic parameters, we compare the resistance to exogenous shocks in infection rates of the ``autarkic'' system, in which the two countries are assumed not to trade with each other, with the resistance of the ``globalized'' system where, instead, cross-country trade is allowed. Overall, globalized systems result more ``extreme'' in their reaction to shocks with respect to autarkic systems. This is a consequence of the two countries being connected: on the one hand, the globalized system has a larger ``recovery capacity'' when facing relatively small shocks but, on the other hand, it has a larger area where both countries end up being completely infected. In particular, the main possibility which is precluded to globalized systems with respect to autarkic ones is a situation in which only one country is infected while the other is not.
On the contrary, ``autarkic'' systems offer a wider spectrum of possible outcomes resulting from infection shocks and, in particular, they exhibit partial endemic equilibria in which only one location is fully infected while the other is disease free.

By comparing how an autarkic system and a globalized system behave in response to shocks, we are able to understand their similarities and differences. The main result of this shock-resistance analysis is that the behavior of the two systems is substantially different especially when they are subjected to ``1-dimensional large shocks'': when infection shocks hit mainly one location (and only slightly the other), a globalized system either fully recovers or becomes fully infected, while an autarkic system could exhibit partial endemic equilibria, if exposed to the same shock. Depending on the amount of resources allocated to recovery, as measured by the quarantine level $q$ in our framework, a globalized system may be preferable when large resources for quarantine are available, whereas an autarkic system is preferable in case of low resources.

\bibliographystyle{chicago}
\bibliography{bibfile}

\setcounter{section}{1}

\appendix

\global\long\def\thesection{Appendix \Alph{section}}
\global\long\def\thesubsection{\Alph{section}.\arabic{subsection}}

\setcounter{theorem}{0} \global\long\def\thetheorem{\Alph{section}.\arabic{theorem}}

\section{More on the econometric analysis}
\label{sec: econometrics_appendix}

In this section we investigate whether the significant effect of the dummy $\text{Positive}$, observed in Table \ref{tab: empres}, could be a result of a selection process. To this aim, we have estimated a bivariate selection model by maximum likelihood estimation where the main equation is a Tobit model with distance as the dependent variable. The participation equation is a Probit, and estimates the probability of being active (i.e. sending at least one bovine) in quarter $t$.

Although the dependent variable in the main equation is -- when not censored -- continuous, the  identification of the model could depend only on distributional assumptions. For this reason, we have added an exclusion restriction in the participation equation, using data on rainfalls provided by the Italian Air Force (\textit{Centro Operativo Dati per la Meteorologia}). For each municipality, we have imputed the level of rainfalls and its deviation from its quarterly mean by averaging the three closest meteorological stations.\footnote{The meteo stations are around 115 with daily data covering the entire Italian territory.} We have thus included as a regressor in the participation equation the lagged value of the deviation of rainfalls from quarterly mean. Since reduced rainfalls at $t-1$ -- through a negative effect on the production of crops used for animal feed (hay, corn, etc.) -- lower the inflow of bovines in that municipality, this, in turn, is expected to decrease outflows at time $t$.

The bivariate model has been estimated using the \textsc{Stata}\textsuperscript{\textregistered} command \textit{cmp} developed by David Roodman.\footnote{See \cite{roodman2011fitting} for details.}

The coefficient of the dummy $\text{Positive}_{i,t-1}$ indicates that farms with a sick bovine at $t-1$ are less likely to be active at time $t-1$.
The deviation of rainfalls from quarterly mean has the expected positive and  statistically significant effect on the probability of sending cattle at time $t$.
The  $\rho$ coefficient, which estimates the correlation between error terms is negative and significant at 10\% , thus suggesting the presence of a weak negative selection effect. The estimated effect of $\text{Positive}_{i,t-1}$ in the Tobit main equation is, however, very close to the result shown in column 3 of Table \ref{tab: empres}.

\begin{table}[h]
	\caption{Bivariate selection model}
	\label{tab: biTobit}
	\centering
	\begin{tabular}{l@{\hspace*{10mm}}c@{\hspace*{10mm}}c}
		\hline\hline \vspace{-4mm}\\
		& Tobit         & Probit         			\\
		& Distance 		& Pr. active at time $t$   	\\
		\hline \vspace{-4mm}\\
		$\text{Positive}_{i,t-1}$   & 19.462***    	& -0.881*** 			\\
		& -5.591        & (0.047)   \vspace{2mm}\\
		$\text{Stock}_{it}$         & 0.0845***     & 0.0004***  \\
		& (0.001)       & (0.000)	\vspace{2mm}\\
		$\text{Rain\ Dev.\ from\ Mean}_{i,t-1}$ &    & 0.002***  \\
		&               & (0.001)   \vspace{2mm}\\
		Constant               		& 13.444***   	& 1.753***  \\
		& -1.104        & (0.012)   \vspace{2mm}\\
		$\sigma$                    &  90.523***    &            \\
		&  (0.0528)     &           \vspace{2mm}\\
		$\rho_{12}$ 				& \multicolumn{2}{c}{-0.0077*}	\\
		& \multicolumn{2}{c}{(0.0047)}\\
		\hline \vspace{-4mm}\\
		Observations           & 2,267,463     & 2,267,463 \\
		Log likelihood & \multicolumn{2}{c}{-10,207,407 }\\
		\hline\hline
	\end{tabular}
	\caption*{\footnotesize The bivariate Tobit/Probit model has been estimated using the \textsc{Stata}\textsuperscript{\textregistered} command \textit{cmp}. The regression includes time and regional effects. Standard errors clustered at the farm level are shown in parenthesis. Asterisks mean: *** significant at 1\%, ** significant at 5\%,* significant at 10\%.}
\end{table}

\newpage

 \section{Proofs for Sections \ref{sec_1island}, \ref{sec_model} and \ref{sec_discussion}}
 \label{sec: proofs}

 
 \begin{proof}[Proof of Proposition \ref{prop_equilibria_1location}]
 	The derivative $\frac{\de x}{\de t}$, which is a cubic function of $x$, has only three roots $x=0$, $x=q$ and $x=1$, where it becomes equal to 0. Moreover, it is strictly negative when $x \in (0,q)$ and strictly positive when $x \in (q,1)$.
 \end{proof}

 \begin{proof}[Proof of Proposition \ref{prop_system_well_defined}]
 	We want to show that the unit square $[0,1]^2$ is an invariant set under the dynamics defined by system (\ref{SIcubic}). In order to do that, we need to the vector field defining the system of equation, i.e. the right-hand side of (\ref{SIcubic}) as 2-dimensional function of $(x_A,x_B)$ is ``pointing toward the interior'' of the square, while restricted on the borders of it. More formally:
 	\begin{itemize}
 		\item	suppose that $x_A=0$. Then $\dot{x}_A = \nu_A (1-F_A)x_B F_B \geq 0$, for any $x_B \in [0,1]$, as wanted.
 		\item	Suppose, instead, that $x_A = 1$. By assumption, we have that $F_A = 1$ when $x_A = 1$, then
 		$$
 		\dot{x}_A = \nu_A (1-F_A)(1-x_B)F_B - x_A F_A = -1 < 0,
 		$$
 		as we wanted.
 	\end{itemize}
 	An analogous and symmetric reasoning shows that $\dot{x}_B \geq 0$, when $x_B = 0$, and that $\dot{x}_B \leq 0$, when $x_B = 1$.
 \end{proof}

 \begin{proof}[Proof of Proposition \ref{proposition_autarky}]
 	The vector field defining system (\ref{ODEautarky}) is of the form $\mathbf F(x_A,x_B) = (F_A(x_A), F_B(x_B))$, where $F_A(x) = F_B(x) = f(x) := \nu x(1-x)(x-q)$. Then, clearly the system is symmetric with respect to the diagonal, that is $\mathbf F(x_B,x_A) = \left(F_B(x_A,x_B), F_A(x_A,x_B)\right)$.
 	
 	Since $f(x) = 0$ if and only if $x=0$ or $x=q$ or $x=1$, then the equilibria of system (\ref{ODEautarky}) are: $(0,0)$, $(0,1)$, $(1,0)$, $(1,1)$, $(q,q)$, $(0,q)$, $(1,q)$, $(q,0)$ and $(q,1)$. Moreover, since $\dot{x}_A = 0$ when $x_A = 0$, this means that the line $x_A = 0$ in $\R^2$ cannot be crossed by the trajectories of the system. Analogously, the lines $x_A = q$, $x_A = 1$, $x_B = 0$, $x_B = q$, $x_B = 1$ cannot be crossed, which implies that they are invariant and that the unit square $[0,1]^2$ is also invariant under the dynamics defined by system (\ref{ODEautarky}).
 	
 	In order to evaluate the stability of such equilibria, it suffices to study the Jacobian of the system. Now, since the Jacobian is of the form
 	$$
 	\begin{pmatrix}
 	f'(x_A)	&	0		\\
 	0		& f'(x_B) 	\\
 	\end{pmatrix},
 	$$
 	where $f'(x) = \nu [(2-3x)x + q(2x-1)]$, and that $f'(0) = -\nu q < 0$, $f'(q) = \nu q (1-q) > 0$ and $f'(1) = -\nu (1-q)$ when $\nu, q \in (0,1)$, then the study of its eigenvalues simply says that: $(0,0)$, $(0,1)$, $(1,0)$ and $(1,1)$ are asymptotically stable, because the eigenvalues are both negative. The points $(0,q)$, $(q,0)$, $(1,q)$, $(q,1)$ are saddle point because they have eigenvalues of different sign. Lastly, $(q,q)$ is an unstable source point because both its eigenvalues are positive.
 \end{proof}

\section{Analysis of the linear case}
\label{app_linear}

We here study system \eqref{ODEmax}, which comes from the assumptions of agents' linear utility and uniform cost distributions.
In principle, the system is well defined in $\R^2$, but we will restrict our analysis to the unit square $(x_A,x_B) \in [0,1]^2$, in which the fractions of infected agents make sense.
It is continuously differentiable everywhere but the diagonal of $\R^2$, i.e. over $\R^2\setminus \{(x_A,x_B)\in\R^2:x_A=x_B\}$. 
However, thanks to the symmetry of the system guaranteed by the assumptions made, we can separate the analysis focusing on three different parts: the diagonal, the super-diagonal set and the sub-diagonal.
This allows us to use an \textit{ad hoc} strategy to obtain some explicit results. There are two asymptotically stable equilibria, $(x_A,x_B)=(1,1)$ and $(0,0)$, the first corresponding to both countries being fully infected, while the second to both being disease free. 
There is a third equilibrium, $(q,q)$ which is an unstable saddle point. 
Its separatrix curves separate the basins of attraction of the asymptotically stable states, as depicted in Figure \ref{basinattraction}. It is worth noting, though, that they are not explicitly characterizable.\footnote{There is no known way to analytically determine these curves, even in simple dynamical systems. Progresses have been made with their numerical approximations \citep{cavoretto2011approximation}.}

For ease of exposition, let us re-write system \eqref{ODEmax} in vector notation as follows:
\begin{equation}
\label{eq: ODEmax_vector}
\frac{\de}{\de t} (x_A,x_B) = \mathbf V(x_A,x_B),
\end{equation}
where $\mathbf V(x_A,x_B) := \left(V_A(x_A,x_B),V_B(x_A,x_B)\right)$ for all $(x_A,x_B) \in \R^2$ and $V_A$, $V_B$ are the 2-variable, real-valued functions defined respectively by the first and second row of \eqref{ODEmax}. Let us also denote the diagonal by $D := \{(x_A,x_B) \in \R^2 :\, x_A = x_B \}$, and the sets above and below the diagonal respectively by $\Delta^+ := \{(x_A,x_B) \in \R^2 :\, x_A < x_B \}$ and $\Delta^- := \{(x_A,x_B) \in \R^2 :\, x_A > x_B \}$.

\begin{lemma}
\label{lemma_V_symmetric}
	The vector field $\mathbf V$ is symmetric with respect to the diagonal $D$, that is, for all $(x_A,x_B) \in \R^2$:
	$$
	\mathbf V(x_B,x_A) \equiv \left(V_B(x_A,x_B),V_A(x_A,x_B)\right).
	$$
\end{lemma}
\begin{proof}
	The proof follows directly from the definition of $\mathbf V$.
\end{proof}

The vector field $\mathbf V$ can be seen as constituted by three basic pieces, all of which are defined over the entire $\R^2$ but such that they coincide with $\mathbf V$ itself when appropriately restricted on the sets $D$, $\Delta^+$ and $\Delta^-$. The following lemma formalizes this idea.

\begin{lemma}
\label{lemma_V_three_pieces}
	\ 
	\begin{enumerate}
	\item	The vector field $\mathbf V$ when restricted on the diagonal $D$ coincides with 
	$$
	\mathbf V^D(x_A, x_B) := 
	\begin{pmatrix}
	\nu x_A (1 - x_A)(x_A - q) \\
	\nu x_B (1 - x_B)(x_B - q)\\
	\end{pmatrix},
	$$
	which, in turn, is well defined over $\R^2$.
	
	\item 	The vector field $\mathbf V$ when restricted on $\Delta^-$ coincides with
	$$
	\mathbf V^-(x_A, x_B) := 
	\begin{pmatrix}
	\nu (1-x_A+x_B) \Big[ x_A(1-x_A)(x_A-q)(1-x_A+x_B) \Big] - x_A (x_A-x_B) \\
	\nu \Big[ x_B(1-x_B)(x_B-q) + (x_A+x_B-2x_Ax_B)(x_A-x_B) \Big]
	\end{pmatrix}.
	$$
	
	\item 	The vector field $\mathbf V$ when restricted on $\Delta^+$ coincides with
	$$
	\mathbf V^+(x_A, x_B) := 
	\begin{pmatrix}
	\nu \Big[ x_A(1-x_A)(x_A-q) + (x_A+x_B-2x_Ax_B)(x_B-x_A) \Big] \\
	\nu (1-x_B+x_A) \Big[ x_B(1-x_B)(x_B-q)(1-x_B+x_A) \Big] - x_B (x_B-x_A)
	\end{pmatrix}.
	$$
	\end{enumerate}
\end{lemma}
\begin{proof}
	When $(x_A,x_B) \in D$, then $\max\{0,x_A - x_B\} = \max\{0,x_B - x_A\} = 0$. From this, the first point follows from the computation of $\mathbf V$ as defined by \eqref{ODEmax}.
	
	The second point follows because when $(x_A,X_B) \in \Delta^-$, then $\max\{0,x_A - x_B\} = x_A - x_B$ while $\max\{0,x_B - x_A\} = 0$. Analogously for the third point.
\end{proof}

\begin{proposition}
\label{proposition_system_well_defined}
	System \eqref{ODEmax} is symmetric with respect to the diagonal and it is well defined in $\R^2$. The diagonal $D$, the sets $\Delta^+$ and $\Delta^-$ are all invariant with respect to the dynamics defined by system \eqref{ODEmax}.
\end{proposition}
\begin{proof}
	The symmetry of system \eqref{ODEmax} in $\R^2$ follows from that of $\mathbf V$. This implies that the diagonal $D$ has to be invariant and, consequently, also $\Delta^+$, $\Delta^-$ have to be invariant.
	
	Because of the invariance, it suffices to show that the system is well defined when restricted on each of $D$, $\Delta^+$ and $\Delta^-$. From the previous Lemma \ref{lemma_V_three_pieces}, it follows that the system is well defined because $\mathbf V^{D}$, $\mathbf V^-$ and $\mathbf V^+$ are smooth on $\R^2$ and, in particular, on $D$, $\Delta^-$ and $\Delta^+$ respectively.
\end{proof}

\begin{proposition}
\label{proposition_unit_square_invariant}
	The unit square $[0,1]^2$ is invariant with respect to the dynamics defined by system \eqref{ODEmax}. 
\end{proposition}
\begin{proof}
	The following Lemma \ref{lemma_V_signs} implies that, on the borders of the unit square, the vector field $\mathbf V$ points towards the interior.
\end{proof}

\begin{lemma}
\label{lemma_V_signs}
	$V_B(\cdot,\cdot) > 0$ on the segment $(0,1)\times \{0\}$ and $V_A < 0$ on the segment $\{1\} \times (0,1)$. Consequently, by symmetry, $V_B < 0$ on $(0,1)\times \{1\}$ and $V_A > 0$ on $\{0\} \times (0,1)$.
\end{lemma}
\begin{proof}
	From the definition, it follows that for all $x_A \in (0,1)$, it holds that $V_B(x_A, 0) = \nu x_A^2 > 0$. Analogously, for all $x_B \in (0,1)$, it holds that $V_A(1,x_B) = -(1-x_B) < 0$.
\end{proof}

Now we can show that system \eqref{ODEmax} has 3 critical points, specifically with  $(q,q)$ being a saddle point, whose separatrix curves naturally form the boundaries of the basins of attraction of the asymptotically stable points $(0,0)$ and $(1,1)$.

\begin{proposition} 
	\label{prop_app2}
	
	System \emph{(\ref{ODEmax})} has (only) three equilibria:
	\begin{itemize}
		\item	$(x_A,x_B)=(0,0)$ and $(1,1)$, which are asymptotically stable states;
		\item	$(x_A,x_B)=(q,q)$, which is an unstable saddle point.
	\end{itemize}
	Moreover, the two separatrix curves of the saddle $(q,q)$ are such that the unstable one coincide with the diagonal of the square $\{(x_A,x_B)\in [0,1]^2: x_A = x_B\}$, while the stable separatrix are part of the boundary of the basins of attraction of the stable equilibria.
\end{proposition}
\begin{proof}
	The proof follows directly by using Lemma \ref{lemma_V_three_pieces} and then applying Lemma \ref{lemma_stability_pieces}.
\end{proof}

\begin{lemma}
	\label{lemma_jacobian}
	Consider $\mathbf V^-$, as defined in Lemma \ref{lemma_V_three_pieces}. Then its Jacobian $Jac^-(x_A,x_B) := \left(\frac{\partial \mathbf V^-(x_A,x_B)}{\partial x_A} | \frac{\partial \mathbf V^-(x_A,x_B)}{\partial x_B}\right)$ when evaluated:
	\begin{itemize}
		\item	in $(0,0)$, it has both eigenvalues equal to $-q\nu < 0$, for all $q, \nu \in (0,1)$;
		\item	in $(1,1)$, it has eigenvalues equal to $-(1-q)\nu$ and $-1-(1-q)\nu$, which are both negative for all $q, \nu \in (0,1)$;
		\item	in $(q,q)$, it has eigenvalues equal to $q\nu(1-q) > 0$, $-q(1+\nu(1-q))<0$ and corresponding eigenvectors equal to $(1,1)$ and $\left(\frac{1}{-2(1-q)\nu}, 1 \right)$.
	\end{itemize}
	Consider $\mathbf V^+$. Analogously, its Jacobian has negative eigenvalues when evaluated in the points $(0,0)$ and $(1,1)$. Whereas, when evaluated in $(q,q)$, the eigenvalues are the same as above, $q\nu(1-q) > 0$ and $-q(1+\nu(1-q))<0$, but their corresponding eigenvectors are $(1,1)$ and $\left(-2(1-q)\nu, 1 \right)$.
\end{lemma}
\begin{proof}
	By definition of $\mathbf V^-$ in Lemma \ref{lemma_V_three_pieces}, the computation of the derivatives in the point $(q,q)$ yields:
	$$
	Jac^-(q,q) = 
	\begin{pmatrix}
	-q(1-\nu(1-q))	&	q			\\
	2(1-q)q\nu		&	-(1-q)q\nu	\\
	\end{pmatrix}.
	$$
	The eigenvalues and eigenvectors of this matrix are easily computed. 
	
	Analogously, the computation in the point $(0,0)$ and $(1,1)$ gives:
	$$
	Jac^-(0,0) =
	\begin{pmatrix}
	-q\nu	&	0		\\
	0		&	-q\nu	\\
	\end{pmatrix},
	\quad
	Jac^-(1,1) =
	\begin{pmatrix}
	-1-(1-q)\nu	&	1			\\
	0			&	-(1-q)\nu	\\
	\end{pmatrix},
	$$
	from which one obtains the eigenvalues and eigenvectors as claimed above.
	
	The last part follows from the same computations done symmetrically for $\mathbf V^+$.
\end{proof}

\begin{lemma}
\label{lemma_stability_pieces}
	Consider the two vector fields $\mathbf V^-$ and $\mathbf V^+$ defined in Lemma \ref{lemma_V_three_pieces}. Then:
	\begin{enumerate}
		\item	The points $(0,0)$, $(1,1)$ and $(q,q)$ are the equilibria of both $\mathbf V^-$ and $\mathbf V^+$ respectively in the region $D\cup\Delta^-$ and $D\cup\Delta^+$.
		\item 	$(0,0)$ and $(1,1)$ are asymptotically stable for both $\mathbf V^-$ and $\mathbf V^+$.
		\item 	$(q,q)$ is a saddle for both $\mathbf V^-$ and $\mathbf V^+$.
	\end{enumerate}
\end{lemma}
\begin{proof}\ 
	For the first point, consider $\mathbf V^-$. It suffices to verify that, in the region $D\cup\Delta^-$, $\mathbf V^-(x_A, x_B) = \mathbf 0$ if and only if $(x_A,x_B)$ is one of the points considered in the claim. Analogously, for  $\mathbf V^+$.
	
	The second and third points follow directly from Lemma \ref{lemma_jacobian}.
\end{proof}

Lastly, we focus on the crucial role played by the stable separatrix curve of the saddle $(q,q)$, hereafter denoted by $\mathcal C$. From the theory of dynamical systems, it follows that $\mathcal C$ is partitioned as the image of three distinct trajectories/solutions of system (\ref{ODEmax}):
$$
\mathcal C = \mathcal C^- \cup \{q,q\} \cup \mathcal C^+.
$$
In our case, $\mathcal C^-$ is the piece obtained as the separatrix of the saddle $(q,q)$ with respect to the vector field $\mathbf V^-$, while  $\mathcal C^+$ is the piece obtained as the separatrix of $(q,q)$ with respect to $\mathbf V^+$.

The following result formalizes what is shown in Figure \ref{basinattraction}: depending on the parameters $\nu$ and $q$, as time $t$ increases, the solution $\mathcal C^-$ enters the unit square either crossing its border along the segment $[q,1]\times \{0\}$ or along $\{1\} \times [0,q]$ and, eventually, converges toward $\{q,q\}$ as $t\rightarrow \infty$. Symmetrically, the same occurs for $\mathcal C^+$.

\begin{proposition}\label{proposition_separatrix}
	Let $\mathcal C$ denote the (unique) stable separatrix of the saddle point $(q,q)$ of system \emph{(\ref{ODEmax})}. The curve $\mathcal C$ can be naturally partitioned according to the following three distinct trajectories that compose it:
	$$
	\mathcal C = \mathcal C^- \cup \{q,q\} \cup \mathcal C^+.
	$$
	Then $\mathcal C^- \cap [0,1]^2$ is included in $[q,1]\times[0,q]$ and, depending on the parameters $q$, $\nu$, it either crosses the segment $[q,1]\times \{0\}$ in a point $(\eta,0)$ or the segment $\{1\}\times [0,q]$ in a point $(1,\zeta)$. A symmetric result holds for $\mathcal C^+$.
\end{proposition}
\begin{proof}
The result is based on the following lemmas.
\end{proof}

\begin{lemma}
	Let $\mathcal C^-$ denote the part of the (stable) separatrix curve of $(q,q)$ with respect to $\mathbf V^-$ that belongs to $\Delta^-$. Symmetrically, let $\mathcal C^+$ denote the (stable) separatrix of $(q,q)$ for $\mathbf V^+$ belonging to $\Delta^+$. Then: $\mathcal C^- \subset [q,1]\times [0,q]$ and $\mathcal C^+ \subset [0,q]\times [q,1]$ for times $t$ large enough.
\end{lemma}
\begin{proof}
	Consider the case of $\mathcal C^-$ (the other case of $\mathcal C^+$ is symmetrical). The point $(q,q)$ is a saddle so its stable separatrix converges to $(q,q)$ as $t\rightarrow\infty$ and, additionally, it is locally linearly approximated by the vector $\left(\frac{1}{-2(1-q)\nu}, 1 \right)$, which is the eigenvector corresponding of the negative eigenvalue of $Jac^-(q,q)$, from Lemma \ref{lemma_jacobian}.
\end{proof}

\begin{lemma}
	The signs of the components of the vector field defining system \emph{(\ref{ODEmax})}, when computed in $(x_A,x_B) \in (q,1)\times(0,q)$, are such that $V_A(x_A,x_B) \leq 0$ and $V_B(x_A,x_B) \geq 0$.
	
	Symmetrically, $V_A \geq 0$ and $V_B \leq 0$ in $(0,q)\times(q,1)$.
\end{lemma}
\begin{proof}
	Consider $V^-_A(x_A,x_B)$, the first component of $\mathbf V^-$ (the other cases are analogous), and let us show that $V^-_A(x_A,x_B) < 0$ when $0 < x_B < q < x_A < 1$. By definition in Lemma \ref{lemma_V_three_pieces}:
	$$
	V^-_A(x_A,x_B) = \nu (1-x_A + x_B)^2 [x_A(1-x_A)(x_A-q)] - x_A(x_A-x_B).
	$$
	It has to be shown that for $0 < x_B < q < x_A < 1$ and $q,\nu \in (0,1)$
	$$
	\nu (1 - x_A + x_B)^2 x_A (1-x_A) (x_A - q) \stackrel{?}{<} x_A (x_A - x_B),
	$$
	that is
	$$
	\nu (1 - x_A + x_B)^2 (1-x_A) (x_A - q) \stackrel{?}{<} x_A - x_B.
	$$
	The right-hand side is always greater than $x_A - q$, so that the inequality becomes:
	$$
	\nu (1 - x_A + x_B)^2 (1-x_A) (x_A - q) \stackrel{?}{<} x_A - q,
	$$
	and then
	$$
	\nu (1 - x_A + x_B)^2 (1-x_A) \stackrel{?}{<} 1.
	$$
	Taking the supremum of the left-hand side for $x_B \in (0,q)$, which is attained at $x_B = q$ gives
	$$
	\nu (1 - x_A + q)^2 (1-x_A) \stackrel{?}{<} 1.
	$$
	Then the supremum for $x_A \in (q,1)$, obtained for $x_A = q$ gives
	$$
	\nu (1 - q + q)^2 (1-q) \equiv \nu (1-q) < 1,
	$$
	which implies that all the inequalities above have to hold, as wanted.
\end{proof}

\section{Linearization of the separatrix curve $\mathcal C$ and approximation of the basins of attraction}
\label{app_comparative}

Since that we have already observed that the separatrix $\mathcal C$ cannot be described analytically, we first compute a linear approximation of it and, then, confirm the results by numerical analysis. This allows us to approximate the area of the basins of attraction which is key for the comparative statics analysis done in Section \ref{sec_discussion}.

We linearize system \eqref{ODEmax} in a neighborhood of the saddle $(q,q)$ using Lemma \ref{lemma_jacobian}. Figure \ref{linearseparatrix} shows similarities and differences between $\mathcal C$ and its linear approximation $\widetilde{\mathcal C}$.

\begin{proposition}
\label{lemmalinearseparatrix}
	The separatrix $\mathcal C$ is linearly approximated in $(q,q)$ by the two-piece line $\widetilde{\mathcal C}$, which we respectively call $\widetilde{\mathcal C}^+$ and $\widetilde{\mathcal C}^-$, defined by
	$$
	\widetilde{\mathcal C} =
	\begin{cases}
	\widetilde{\mathcal C}^+: \quad x_B = \dfrac{1}{-2(1-q)\nu}(x_A-q) + q,	&	\text{defined for } x_A \leq q, 	\vspace{.2cm}\\
	\widetilde{\mathcal C}^-: \quad x_B = -2(1-q)\nu (x_A-q) + q,			&	\text{defined for } x_A \geq q.
	\end{cases}
	$$
\end{proposition}
\begin{proof}
	Let us consider $\widetilde{\mathcal C}^-$ (the case of $\widetilde{\mathcal C}^+$ is analogous). From Lemma \ref{lemma_jacobian} it follows that the approximation of the (stable) separatrix of $(q,q)$ with respect to $\mathbf V^-$ is the line for $(q,q)$ with tangent given by the eigenvector corresponding to the negative eigenvalue, that is the vector $\left(\frac{1}{-2(1-q)\nu}, 1 \right)$. Such line in $\R^2$ is parametrically described by
	$$
	\begin{pmatrix}
	x_A		\\
	x_B
	\end{pmatrix} = 
	t
	\begin{pmatrix}
	\frac{1}{-2(1-q)\nu}	\\
	1
	\end{pmatrix} 
	+
	\begin{pmatrix}
	q	\\
	q
	\end{pmatrix}, 
	\quad
	\text{for } t \geq 0
	$$
	which is implicitly written as $x_B = -2(1-q)\nu (x_A-q) + q$, for $x_A \geq q$.
\end{proof}

\begin{lemma}
\label{lemma_point_intersection}
	The intersection between $\widetilde{\mathcal C}^-$ and the sub-diagonal boundaries of the unit square $[0,1]^2$, that is, the segments $\{1\} \times [0,q]$ and $[0,q] \times \{0\}$, is the point $P^- = (P^-_A,P^-_B)$ given by
	$$
	P^- = 
	\left\{
	\begin{array}{rr}
	\left(1, -2\nu(1-q)^2 + q \right),			&	\text{if } \nu < \frac{q}{2(1-q)^2},	\\
	\left(\frac{q}{2\nu(1-q)} + q, 0 \right), 	&	\text{if } \nu \geq \frac{q}{2(1-q)^2},
	\end{array}
	\right.
	$$
	Symmetrically, a point $P^+$ can be found as the intersection of $\widetilde{\mathcal C}^+$ and the segments $\{0\}\times [q,1]$ and $[0,q]\times \{1\}$.
\end{lemma}

\begin{remark}
Notice that, provided $q \in (0,1)$ and $\nu \in (0,1)$, the two following conditions are equivalent:
$$
\nu < \frac{q}{2(1-q)^2} \quad \Longleftrightarrow \quad  \frac{1 + 4 \nu -\sqrt{8 \nu +1}}{4 \nu } < q.
$$
Moreover, when $q > 1/2$ then $\nu < \frac{q}{2(1-q)^2}$ for all $\nu \in (0,1)$. Figure \ref{conditiononqv} shows the subregion of the square $(q,\nu)\in(0,1)^2$ where this condition is satisfied.
\end{remark}

\begin{figure}[htbp]
	\centering
	\caption{\textbf{Condition on the parameters $q$ and $\nu$}}
	\begin{minipage}[c]{.5\textwidth}
		\centering\setlength{\captionmargin}{0pt}
		\includegraphics[width=\textwidth]{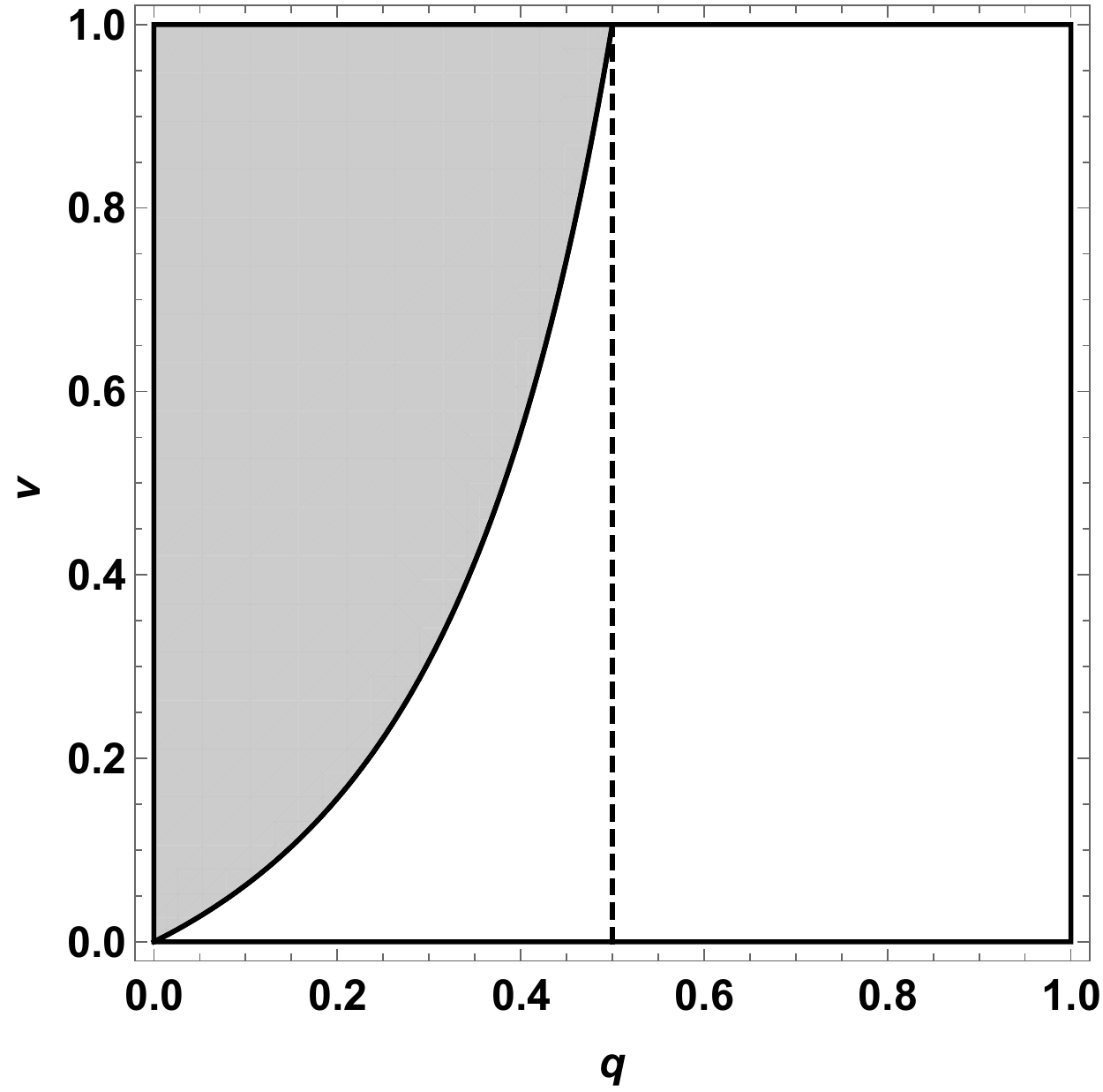}
	\end{minipage}
	\hfill
	\begin{minipage}[c]{.47\textwidth}
		\centering\setlength{\captionmargin}{0pt}
		\caption*{\small Subregions of the square $(q,\nu) \in (0,1)^2$ separated by the curve $\nu = \frac{q}{2(1-q)^2}$. The white area is where $\nu < \frac{q}{2(1-q)^2}$, whereas the gray area is where the opposite inequality holds. In particular, in the white area (respectively, gray) $P^-$ belongs to the vertical segment $\{1\}\times [0,q]$ (resp. horizontal segment $[q,1]\times \{0\}$) and the area under the curve $\widetilde{\mathcal C}^-$ is the trapezoid $Q^-$ (resp. triangle $T^-$). Lastly, the dashed line is at $q=1/2$.}
	\end{minipage}
	\label{conditiononqv}
\end{figure}

Depending on the parameters $\nu$ and $q$, the area under the curve $\widetilde{\mathcal C}$ is either a trapezoid or a triangle and is easily computed in the following result. By considering this area as an approximation of the area under the curve $\mathcal C$, this will also allow us to make a comparative statics analysis. The results are also shown in Figure \ref{approximatedarea}.

\begin{lemma}
\label{lemmatriangletrapezoid}
	If $\nu \geq \frac{q}{2(1-q)^2}$, consider the triangle $T^-\subset\{(x_A,x_B)\in[0,1]^2:x_A\geq x_B\}$ defined as the convex hull in $\R^2$ of the following set of vertexes 
	$$
	T^- = \mathrm{Conv}\left(\{(q,q), (q,0), P^-\}\right).
	$$
	If, instead, $\nu < \frac{q}{2(1-q)^2}$, consider the trapezoid $Q^-\subset\{(x_A,x_B)\in[0,1]^2:x_A\geq x_B\}$ defined by
	$$
	Q^- = \mathrm{Conv}\left(\{(q,q), (q,0), (1,0), P^-\}\right).
	$$
	The measure of their area is:
	$$
	\begin{aligned}
	& \mathrm{A}(T^-) = \dfrac{q \times (P^-_A - q)}{2} = \dfrac{q^2}{4\nu(1-q)}, \quad\text{defined whenever } \nu \geq \dfrac{q}{2(1-q)^2},\\
	& \mathrm{A}(Q^-) = \dfrac{(1-q) \times (q + P^-_B)}{2} = (1-q)  \left(q -\nu(1-q)^2\right), \quad\text{when } \nu < \dfrac{q}{2(1-q)^2}.
	\end{aligned}
	$$
	Whenever defined, $q \mapsto \left[\mathrm{A}(T^-)\right](q,\nu)$ is always increasing for all $\nu$. Moreover, its derivative with respect to $q$ is:
	$$
	\frac{\partial \mathrm{A}(T^-)}{\partial q} = \dfrac{q(2-q)}{4\nu(1-q)^2} > 0, \quad\forall \ q,\nu \in (0,1):\ \nu \geq \dfrac{q}{2(1-q)^2}.
	$$
	The derivative of $\mathrm{A}(Q^-)$ is
	$$
	\frac{\partial \mathrm{A}(Q^-)}{\partial q} = 1 - 2q + 3\nu(1-q)^2,\quad\text{defined whenever } \nu < \dfrac{q}{2(1-q)^2},
	$$
	and it is positive if and only if the following condition holds\footnote{The complicated condition derives from the fact that $\frac{1-2q}{3(1-q)^2}$ is decreasing while $\frac{q}{2(1-q)^2}$ is increasing, when $q\in(0,1)$, and they cross in $q=2/7$.}
	\begin{multline*}
	\left\{ \frac{2}{7} < q \quad\land\quad \dfrac{1-2q}{3(1-q)^2} < \nu < \dfrac{q}{2(1-q)^2}\right\} \lor {} \\
	\left\{ q < \frac{2}{7} \quad\land\quad \left[ \nu < \dfrac{q}{2(1-q)^2} \quad\lor\quad \nu > \dfrac{1-2q}{3(1-q)^2} \right] \right\}.
	\end{multline*}
	Analogously, for $P^+$ and correspondingly $T^+$, $Q^+$ and their areas and derivatives.
\end{lemma}
\begin{proof}
	The measure of the areas of the triangle $T^-$ or trapezoid $Q^-$ are easily computed by using the coordinates of $P^-$ obtained in Lemma \ref{lemma_point_intersection}. Computing the derivatives is then straightforward. 
\end{proof}

Now let us compute the ratio between the area under the curve $\widetilde{\mathcal C}^-$ and the entire rectangle $[q,1-q] \times [0,q]$, as in Figure \ref{approximatedarea}.

\begin{lemma}
	Let $\widetilde{R}^-(q,\nu)$ be the ratio between the area under the curve $\widetilde{\mathcal C}^-$ and the rectangle $[q,1] \times [0,q] \subset [0,1]^2$. Then $\widetilde{R}^-(q,\nu)$ is given by
	$$
	\widetilde{R}^-(q,\nu) :=
	\left\{
	\begin{array}{rr}
	\dfrac{\left[\mathrm{A}(T^-)\right](q,\nu)}{q(1-q)} = \dfrac{q}{4\nu (1-q)^2}, & \text{if } \nu \geq \dfrac{q}{2(1-q)^2}	\vspace{.2cm}\\
	\dfrac{\left[\mathrm{A}(T^-)\right](q,\nu)}{q(1-q)} \equiv \dfrac{\left[\mathrm{A}(Q^-)\right](q,\nu)}{q(1-q)} = \frac{1}{2},		&	\text{if } \nu = \dfrac{q}{2(1-q)^2}	\vspace{.2cm}\\
	\dfrac{\left[\mathrm{A}(Q^-)\right](q,\nu)}{q(1-q)} = \dfrac{q - \nu (1-q)^2}{q}, & \text{if } \nu \leq \dfrac{q}{2(1-q)^2}.
	\end{array}
	\right.
	$$
	In an analogous fashion, the ratio above the line $\widetilde{\mathcal C}^+$ is denoted by $\widetilde{R}^+(q,\nu)$ and is equal to $\widetilde{R}^-(q,\nu)$ by symmetry.
\end{lemma}
\begin{proof}
	The numerator of $\widetilde{R}^-$ is given by the area of the triangle or trapezoid given by the previous Lemma \ref{lemmatriangletrapezoid}, that is, respectively $A(T^-)$ or $A(Q^-)$. The denominator, instead, is simply the area of the rectangle $[q,1] \times [0,q]$ in $\R^2$.
\end{proof}

The behavior of $\widetilde{R}(q,\nu)$, as a function of the parameters $q$ and $\nu$, is described by the following result.

\begin{lemma}[\textbf{Comparative statics on the approximated ratio $\widetilde{R}(q,\nu)$}]\ \\
Consider $\widetilde{R}^-(q,\nu)$ defined above. 
It is bounded in $[0,1]$ and its sections $q \mapsto\widetilde{R}^-(q,\nu)$ are increasing for all $\nu \in (0,1)$, whereas $\nu \mapsto\widetilde{R}^-(q,\nu)$ are decreasing for all $q\in(0,1)$. Furthermore, whenever defined, its derivatives are:
$$
\begin{aligned}
& \frac{\partial}{\partial q}\widetilde{R}^-(q,\nu) = 
\begin{cases}
\dfrac{1+q}{4\nu(1-q)^3} > 0,  & \text{if } \nu > \dfrac{q}{2(1-q)^2}	\vspace{.2cm}	\\
\left(\dfrac{1}{q^2}-1\right) \nu > 0, & \text{if } \nu < \dfrac{q}{2(1-q)^2}	\\ 
\end{cases} \\
& \frac{\partial}{\partial \nu}\widetilde{R}^-(q,\nu) =
\begin{cases}
-\dfrac{q}{4 (1-q)^2 \nu^2} < 0, & \text{if } \nu > \dfrac{q}{2(1-q)^2}  \vspace{.2cm}	\\
\dfrac{-(1 - q)^2}{q} < 0, & \text{if } \nu < \dfrac{q}{2(1-q)^2}.		\\
\end{cases}
\end{aligned}
$$
\end{lemma}
\begin{proof}
	The computation of the derivatives follows directly from the formulas defining $\widetilde{R}^-$ in the previous Lemma. Moreover, it is straightforward to check that the denominator is always greater than the numerator, thus guaranteeing that $\widetilde{R}^-(q,\nu) \leq 1$.
\end{proof}

\begin{remark}
	Given $q$ and $\nu$, $R^-(q,\nu)$ just represents the part of the basin of attraction of $(0,0)$ within the rectangle $[q,1]\times[0,q]$. So, the exact total area of the basin of attraction of $(0,0)$ is readily computed as: $2 \cdot R^-(q,\nu) + q^2$. Analogously for its approximation obtained with $\widetilde{R}^-$. These areas (and their differences) are shown in Figure \ref{ratiofunctionofqv}.
\end{remark}

\begin{figure}[tbp]
	\centering
	\caption{\textbf{Area of basin of attraction and its approximation}}
	\includegraphics[width=.7\textwidth]{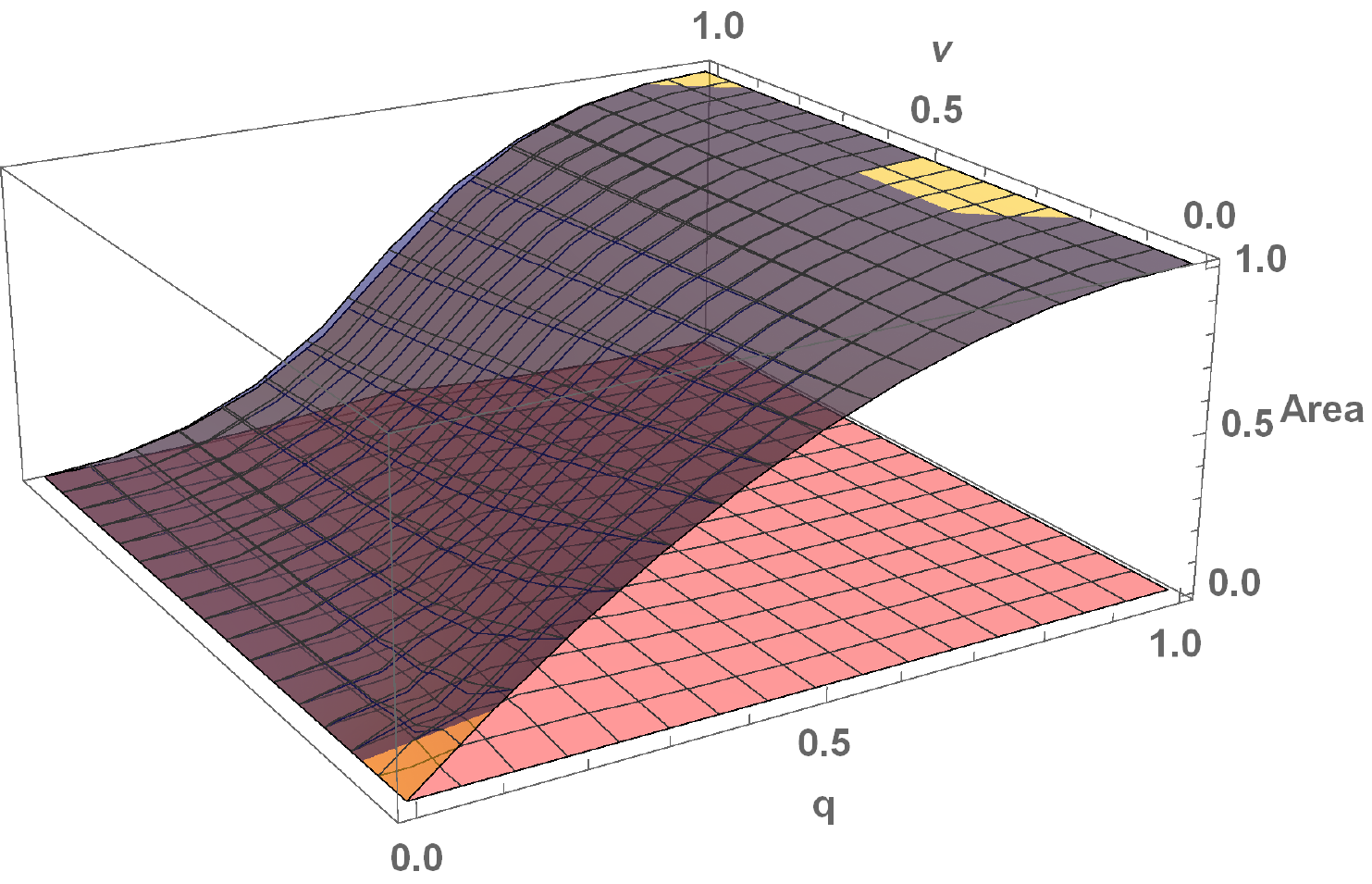}
	\caption*{\small The plot shows: the area of the basin of attraction of $(0,0)$ as a function of $q$ and $\nu$ computed numerically (in blue), the area computed using the approximation $\widetilde{\mathcal C}$ (in yellow, almost indistinguishable from that in blue) and, finally, their differences (the almost-plain surface, in red).}
	\label{ratiofunctionofqv}
\end{figure}

\end{document}